\documentclass[journal]{IEEEtran}
\usepackage{bbm}
\usepackage{amsmath,graphicx,amssymb}
\usepackage{mathtools,dsfont,amsfonts,units,amsthm,bm}
\usepackage{algorithm}
\usepackage{enumerate}
\usepackage[noend]{algpseudocode}
\usepackage{tcolorbox}
\usepackage{graphicx} 
\usepackage{booktabs} 
\usepackage{amsfonts} 
\usepackage{amsthm}
\usepackage{amsmath,graphicx,amssymb}
\usepackage{mathtools,dsfont,amsfonts,units,amsthm,bm}
\usepackage{algorithm}
\usepackage{mdframed}
\usepackage{epstopdf}
\usepackage{subcaption}

\usepackage{hyperref}
\usepackage[numbers]{natbib}

\usepackage{lipsum}

\newboolean{showcomments}
\setboolean{showcomments}{true}
\newcommand{\weiyu}[1]{\ifthenelse{\boolean{showcomments}}
{ \textcolor{blue}{(Weiyu says:  #1)}}{}}
\newcommand{\ehsan}[1]{  \ifthenelse{\boolean{showcomments}}
{ \textcolor{red}{(Ehsan says:  #1)}} {}  }
\newcommand{\chris}[1]{\ifthenelse{\boolean{showcomments}}
{ \textcolor{magenta}{(Chris says: #1)} } {} }
\newcommand{\babak}[1]{\ifthenelse{\boolean{showcomments}}
{ \textcolor{green}{(Babak says:  #1)}}{}}

\title{
Symbol Error Rate Performance of\\ Box-relaxation Decoders\\ in Massive MIMO
}
\author{Christos Thrampoulidis$^\star$, Weiyu Xu$^\dagger$, Babak Hassibi$^\ddagger$
\thanks{
$^\star$Research Laboratory of Electronics, MIT, Cambridge, USA, 
$^\dagger$Department of ECE, University of Iowa, Iowa city, USA, 
$^\ddagger$Department of Electrical Engeeniring, Caltech, Pasadena, USA.
}
}

\definecolor{darkred}{RGB}{250,0,0}
\definecolor{darkgreen}{RGB}{0,150,0}
\definecolor{myblue}{RGB}{0,0,250}
\definecolor{darkblue}{RGB}{0,0,200}
\hypersetup{colorlinks=true, linkcolor=darkred, citecolor=myblue, urlcolor=darkblue}

\newcommand{\psiu}{\overline{\psi}}
\newcommand{\psil}{\underline{\psi}}

\newcommand{\Rb}{\mathbf{R}}
\newcommand{\Qb}{\mathbf{Q}}

\newcommand{\plim}{\operatorname{plim}}
\newcommand{\nlim}{\operatorname{{ n \rightarrow \infty }}}

\newcommand{\plimn}{\underset{\nlim}{\plim}}

\newcommand{\FM}{F_M}

\newcommand{\ksi}{\xi}
\newcommand{\taut}{\tilde{\tau}}

\newcommand{\elli}{\ell_{j}}
\newcommand{\ui}{u_{j}}

\newcommand{\ellix}{\x^-_{0,i}}
\newcommand{\uix}{\x^+_{0,i}}

\newcommand{\wh}{{\hat\w}}

\newcommand{\psio}{\overline{\psi}}

\newcommand{\Pe}{P_e}
\newcommand{\BER}{\mathrm{SER}}
\newcommand{\SER}{\mathrm{SER}}

\usepackage{dsfont}
\newcommand{\ind}[1]{{\mathds{1}}_{\{#1\}}}
\newcommand{\gm}{\|\g\|}
\newcommand{\dB}{dB }

\newcommand{\PeMF}{P_e^{MFB}}

\newcommand{\Pro}{\mathbb{P}}
\newcommand{\psiubw}{\psi(\w,\ub)}
\newcommand{\Scc}{{\Sc^c}}

\newcommand{\phio}{\overline{\phi}}

\theoremstyle{theorem}

\newtheorem{thm}{Theorem}[section]

\newtheorem{lem}{Lemma}[section]

\newtheorem{cor}{Corollary}[section]

\theoremstyle{remark}

\theoremstyle{definition}

\newcommand{\uniform}{\citep[Cor..~II.1]{AG1982}}

\newcommand{\eps}{\epsilon}

\newcommand{\wt}{\tilde\w}

\newcommand{\sign}{\mathrm{sign}}

\newcommand{\SNR}{\mathrm{SNR}}

\newcommand{\E}{\mathbb{E}}                    
\newcommand{\sigg}{\sigma^2}

\newcommand{\nn}{\notag}
\newcommand{\R}{\mathbb{R}}

\newcommand{\G}{\mathbf{G}}

\newcommand{\A}{\mathbf{A}}

\newcommand{\x}{\mathbf{x}}
\newcommand{\w}{\mathbf{w}}
\newcommand{\ub}{\mathbf{u}}
\newcommand{\g}{\mathbf{g}}
\newcommand{\vb}{\mathbf{v}}

\newcommand{\y}{\mathbf{y}}

\newcommand{\z}{\mathbf{z}}

\newcommand{\ab}{\mathbf{a}}

\newcommand{\h}{\mathbf{h}}

\newcommand{\Sc}{{\mathcal{S}}}

\newcommand{\Nn}{\mathcal{N}}

\newcommand{\Cc}{\mathcal{C}}

\newcommand{\Ec}{\mathcal{E}}

\newcommand{\beq}{\begin{equation}}
\newcommand{\eeq}{\end{equation}}
\newcommand{\bea}{\begin{align}}
\newcommand{\eea}{\end{align}}

\newcommand{\vp}{\vspace{4pt}}

\newcommand{\rP}{\xrightarrow{P}}

\begin{document}

\maketitle
\begin{abstract}
The maximum-likelihood (ML) decoder for symbol detection in large multiple-input multiple-output  wireless communication systems is typically computationally prohibitive. In this paper, we study a popular and practical alternative, namely the Box-relaxation optimization (BRO) decoder, which is a natural convex relaxation of the ML. For iid real Gaussian channels with additive Gaussian noise,  we obtain exact asymptotic expressions for the symbol error rate (SER) of the BRO. The formulas are particularly simple, they yield useful insights, and they allow accurate comparisons to the matched-filter bound (MFB) and to the zero-forcing decoder. 
For BPSK signals the SER performance of the BRO is within $3$\dB of the MFB for square systems, and it approaches the MFB as the number of receive antennas grows large compared to the number of transmit antennas.  Our analysis further characterizes the empirical density function of the solution of the BRO, and shows that error events for any fixed number of symbols are asymptotically independent. The fundamental tool behind the analysis is the convex Gaussian min-max theorem.
\end{abstract}



\section{Introduction}
%
%
%
%
The problem of recovering an unknown vector of symbols that belong to a finite constellation from a set of noise corrupted linearly related measurements arises in numerous applications, and in particular in  multiple-input multiple output (MIMO) wireless communication systems \cite{ngo2013energy,rusek2013scaling,wen2014message,narasimhan2014channel}. As a result, a large host of exact and heuristic optimization algorithms have been proposed over the years. Exact algorithms, such as sphere decoding and its variants, become computationally prohibitive as the problem dimension grows, a scenario that is typical in modern massive MIMO systems, e.g., \cite{rusek2013scaling}. Heuristic algorithms such as zero-forcing, MMSE, decision-feedback, etc., \cite{grotschel2012geometric,foschini1996layered,hassibi2005sphere,guo2003multiuser} have inferior performances that are often difficult to precisely characterize.
One popular heuristic is the so called box-relaxation optimization decoder, which is a natural convex relaxation of the maximum-likelihood (ML) decoder, and which allows one to recover the signal via efficient convex optimization followed by hard thresholding, e.g., \cite{tan2001constrained,yener2002cdma,ma2002quasi}. Despite its popularity, very little is known analytically about the decoding performance of this method. In this paper, we close this gap by deriving exact asymptotic error-rate characterizations under the assumption of real Gaussian wireless channel and additive Gaussian noise. 
%


\subsection{Problem formulation}

We consider the problem of recovering an unknown vector $\x_0$ of $n$ transmitted symbols each belonging to a finite constellation from 
the noisy multiple-input multiple-output relation,
$
\y=\A\x_0+\z\in\R^m,$
where $\A\in\R^{m\times n}$ is the MIMO channel matrix (assumed to be known) and $\z\in\R^m$ is the noise vector.
We assume iid real Gaussian channel with additive Gaussian noise. In particular, $\A$ has entries iid $\Nn(0,1/n)$ and $\z$ has entries iid $\Nn(0,\sigma^2)$. The normalization is such that the signal-to-noise ratio (SNR) varies inversely proportional to the noise variance $\sigma^2$. We are interested in the large-system limit, where both the number $n$ of transmit antennas  and the number $m$ of receive antennas are large.
For simplicity of exposition we assume, for the most part of the paper, that $\x_0$ is an n-dimensional BPSK vector, i.e., $\x_0\in\{\pm1\}^n$. Extensions to M-ary constellations are also provided.


\noindent\textbf{Maximum-Likelihood decoder}. The ML decoder for BPSK signal recovery, which maximizes the block error probability (assuming the $\x_{0,i}$ are equally likely) is given by
$
\min_{\x\in\{\pm1\}^n}\|\y-\A\x\|_2.
$
 Solving for the exact ML solution is often computationally intractable, especially when $n$ is large, and therefore a variety of heuristics have been proposed (zero-forcing, mmse, decision-feedback, etc.) \citep{verdu1998multiuser,guo2003multiuser}.

\noindent\textbf{Box-relaxation optimization decoder}.
The heuristic we consider in this paper is the box-relaxation optimization (BRO) decoder \citep{tan2001constrained,yener2002cdma,ma2002quasi}. It consists of two steps. The first one involves solving a convex relaxation of the ML algorithm, where $\x\in\{\pm1\}^n$ is relaxed to $\x\in[-1,1]^n$. The output of the optimization is  hard-thresholded in the second step to produce the final binary estimate. Formally, the algorithm outputs an estimate
$\x^*$
 of $\x_0$ given as
\begin{subequations}\label{eq:algo}
\begin{align}
\hat\x = \arg\min_{-1\leq\x_i\leq 1}\|\y-\A\x\|_2,\label{eq:LASSO}\\
\x^* = \sign(\hat\x),
\end{align}
\end{subequations}
where the $\sign(\cdot)$ function returns the sign of its input and acts element-wise on input vectors.
The BRO decoder naturally extends to the case of recovering signals from higher-order constellations; see Section \ref{sec:MPAM}.

\noindent\textbf{Symbol error rate}.
 We evaluate the performance of the decoder by the symbol error rate (SER),  defined as
 \begin{align}
\SER &:=  \frac{1}{n}\sum_{i=1}^n \ind{\x^*_i\neq \x_{0,i}}, \label{eq:BER}
 \end{align}
 with $\ind{}$ used to denote the indicator function. 
 A closely related quantity that is also of interest is the symbol-error probability $\Pe$, which is  defined as the expectation of the SER averaged over the noise, over the channel, and over the constellation. Formally,
\begin{align}
 \Pe := \E\left[\BER \right] &= \frac{1}{n}\sum_{i=1}^n \Pr\left(\x^*_i\neq \x_{0,i}\right)\label{eq:Pe}.
 \end{align}

\subsection{Contribution and related work}
In this paper, we derive the first rigorous precise characterization of the SER for the BRO decoder in the large-system limit, where the numbers $m$ and $n$ of receive and transmit antennas grow proportionally large at a fixed rate $\delta=m/n$. We complement the precise error formulas with closed-form, tight, upper and lower bounds that are simple functions of the SNR and of $\delta$. These bounds allow useful insights on the decoding performance of the BRO, and they allow a quantitative comparison to the matched-filter bound (MFB) and to the zero-forcing (ZF) decoder. As a concrete example, for BPSK signals the SER of the BRO at high-SNR is $Q(\sqrt{(\delta-1/2)\SNR}),$ where the $Q$-function is the tail probability of the standard normal distribution. This value is within $3$\dB of the MFB for square systems, and it approaches the MFB as $m$ approaches $n$. Finally, we evaluate the large-system empirical distribution of the output of the BRO, and we show that error events for any fixed number of symbol-errors are asymptotically independent. 

To the best of our knowledge, a precise formula for the SER was unknown for the BRO. We remark that the replica method developed in statistical mechanics can be used to give formulas for the SER of various detectors in multiuser detection for code-division multiple access (CDMA) or massive MIMO systems. However, the replica method involves a set of conjectured assumptions that remain mostly unverified by rigorous means; please see \cite{CDMAReplicaTanaka,guo2003multiuser,guo2005randomly} and references therein. In contrast, our analysis is rigorous, and the techniques used are fundamentally different. They are based on recent advances in comparison inequalities for Gaussian processes; in particular, the convex Gaussian min-max theorem \cite{COLT15,Master}. 

The present paper is a significantly extended version of our conference paper in \cite{ICASSP} \footnote{The analysis framework that we present here is general and can be used to analyze the performance of other decoders as well. For example, see our recent papers with co-authors \cite{atitallah2017ber,atitallah2017box}, which build upon the framework of this work.}. In a related recent line of work \cite{jeon2015optimality,ghods2015optimal,jeon2016performance}, the authors have proposed and have investigated the performance of a new class of iterative decoding methods for signal detection in large MIMO systems, which rely on approximate message passing (AMP) \cite{donoho2009message}. The decoding methods that these papers discuss are different than the BRO decoder, and the analysis tools used are also different than the ones presented here. Interestingly, after our paper \cite{ICASSP} appeared, the authors of \cite{jeon2016performance} used our results to show that their proposed algorithm achieves the same error-rate performance as the BRO decoder.

%
%

\vspace{-5pt}
\subsection{Paper Organization}
In Section \ref{sec:result} we analyze the performance of the BRO for BPSK signals. The main theorem of this section, namely Theorem \ref{thm:main}, characterizes the SER and leads to an accurate comparison of the BRO to the MFB and to the ZF decoder. We extend the results to M-PAM constellations in Section \ref{sec:MPAM}. Section \ref{sec:proof} includes the main technical result of the paper, namely Theorem \ref{thm:general}, as well as its detailed proof. The paper concludes in Section \ref{sec:conc} with a discussion on future research directions. Finally, some technical proofs are deferred to the Appendix.

\section{The BRO Decoder for BPSK signals}\label{sec:result}

We precisely analyze the error-rate performance of the BRO decoder for BPSK signals. Our main Theorem \ref{thm:main} in Section \ref{sec:BPSK_main} evaluates its symbol error rate, and simple, closed-form (upper and lower) bounds are computed in Section \ref{sec:BPSK_bounds}. In Sections \ref{sec:BPSK_MFB} and \ref{sec:BPSK_ZF} we use these bounds to compare the BRO decoder to the matched-filter bound, and to the zero-forcing decoder, respectively.

\subsection{Precise $\SER$ performance}\label{sec:BPSK_main}

Our main result explicitly characterizes the limiting behavior of the $\SER$ of the BRO in \eqref{eq:algo}, under a large-system limit in which $m,n\rightarrow+\infty$ at a proportional (constant) rate $\delta>0$. The $\SNR$ is assumed constant; in particular, it does not scale with $n$. Note that for $\x_0\in\{\pm1\}^n$, $\SNR={1}/{\sigma^2}$.

We use standard notation  $\plim_{\nlim}{X_n}=X$  to denote that a sequence of random variables $X_n$ converges \emph{in probability} towards a constant $X$. All limits will be taken in the regime  $m,n\rightarrow+\infty, {m}/{n}=\delta$; to keep notation short we simply write $n\rightarrow\infty$. Finally, we use $Q(\cdot)$ denote the  $Q$-function associated with the standard normal density $p(h)=\frac{1}{\sqrt{2\pi}}\mathrm{e}^{-h^2/2}$.

\begin{thm}[SER for BPSK signals]\label{thm:main}
Let $\SER$ denote the symbol-error-rate of the box-relaxation optimization decoder in \eqref{eq:algo}, for some fixed but unknown BPSK signal $\x_0\in\{\pm1\}^n$. Fix a constant $\SNR$ and a constant $\delta\in(\frac{1}{2},+\infty)$. Then, in the limit of  $m,n\rightarrow+\infty,~{m}/{n}=\delta$, it holds:
$$
\plimn~\SER = Q\left(\frac{1}{\tau_*}\right),
$$
where $\tau_*$ is the unique positive minimizer of the strictly convex function $F:(0,+\infty)\rightarrow\R$ defined as:
\begin{align}\label{eq:F}
F(\tau):={\tau}\big(\delta-\frac{1}{2}\big)+\frac{1/\SNR}{\tau}+\Big(\tau+\frac{4}{\tau}\Big)Q\left(\frac{2}{\tau}\right)-\sqrt{\frac{2}{\pi}}e^{-\frac{2}{\tau^2}}.
\end{align}
\end{thm}

\begin{figure}[t]
    \centering
    \includegraphics[width=0.45\textwidth]{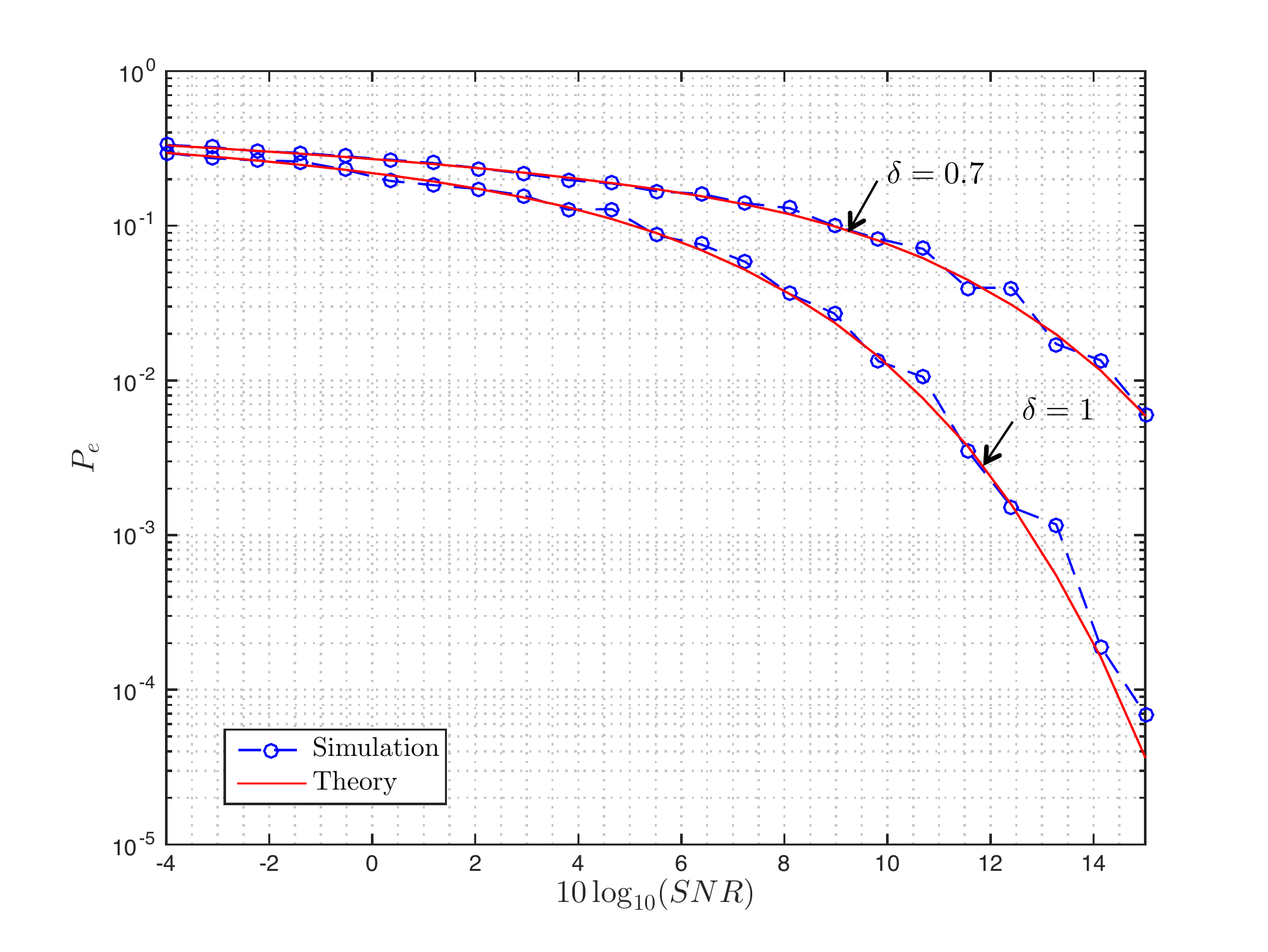}
    \caption{\captionsize{Symbol-error probability of the BRO as a function of $\SNR$ for different values of the ratio $\delta$ of receive to transmit antennas. The theoretical prediction follows from Theorem ~\ref{thm:main}. For the simulations, we used $n=512$. The data are sample averages of the $\SER$ over 20 independent realizations of the channel matrix and of the noise vector for each value of the $\SNR$.}}
    \label{fig:sim}
\end{figure}

The theorem explicitly characterizes the high-probability limit of the SER over the randomness of the channel matrix $\A$, and of the noise vector $\z$. The function $F(\tau)$ in \eqref{eq:F} is deterministic, strictly convex, and is parametrized by the value of the $\SNR$ and by the proportionality factor $\delta$.

The proof of the theorem uses the convex Gaussian min-max theorem (CGMT) \cite{COLT15,Master}, which has thus far found major use in precisely quantifying the squared-error performance of regularized M-estimators in high-dimensions, such as the LASSO \cite{Master}. In this paper we extend the applicability of the CGMT to the characterization of the SER performance, to arrive to Theorem \ref{thm:main}. More than that, along the way we prove a number of even stronger statements regarding the error performance of the BRO. We:
\begin{enumerate}[(i)]
\item establish the large-system error performance of the BRO for a wide class of performance metrics; this class includes the squared-error and the SER as special cases.
\item explicitly characterize the limiting empirical distribution of the output $\hat\x$ of \eqref{eq:LASSO}.
\item show that error events for any fixed number of bits are asymptotically independent.
\end{enumerate}
Please refer to Theorem \ref{thm:general} and to Corollary \ref{cor:indep} for the formal statements of these results. The detailed proof of Theorem \ref{thm:main} is also deferred to Section \ref{sec:proof}.


%

Some further remarks on Theorem \ref{thm:main} are given below.

\vp
\subsubsection{On $\delta>\frac{1}{2}$} The theorem  holds as long as the ratio of proportionality $\delta$ is (strictly) greater than 1/2. To begin with, note that this allows for the number of receive antennas to be less than the number of transmit antennas, and as low as (almost) half of them. When $\delta<1$ the system of linear equations $\y=\A\x$ is underdetermined; hence, recovering  the true solution is generally ill-posed even in the the absence of noise. However, in the problem of interest it is a-priori known that the true solution $\x_0$ only takes values $\{\pm 1\}^n$. The BRO decoder uses that information by enforcing an $\ell_\infty$-norm constraint in \eqref{eq:LASSO}. Of course, this idea of using convex optimization with (typically non-smooth) constraints that promote the particular structure of the unknown signal $\x_0$ to solve underdetermined system of equations,  is one of the core ideas that emerged from the Compressed Sensing literature (e.g. \cite{Cha}). In fact, it is by now well-understood that in the noiseless case the program in \eqref{eq:LASSO} successfully recovers the true $\x_0\in\{\pm1\}^n$ with high probability over the randomness of $\A$ if and only if $\delta>1/2$ (\cite{Cha,TroppEdge}). The same necessary condition naturally arises out of our proof of Theorem \ref{thm:main}.

\vp
\subsubsection{Probability of error}
Recall from \eqref{eq:Pe} that the symbol-error probability is given as $\Pe=\E[\BER]$. Also, the $\BER$ is bounded between $0$ and $1$. Thus, using Theorem \ref{thm:main} 
 we show in Appendix \ref{sec:cor_main_proof} that $\Pe$ converges (deterministically) to the same value $Q(1/\tau_*)$. 

\begin{cor}[$\Pe$]\label{cor:main}
Under the setting of Theorem \ref{thm:main}, let $\Pe$ denote the symbol-error probability  of the BRO and $\tau_*$ be the minimizer of \eqref{eq:F}. Then,
$$
\lim_{\nlim}\Pe = Q\left({1}/{\tau_*}\right).
$$
\end{cor}

\vp
\subsubsection{Solving for $\tau_*$} In order to evaluate the large-system limit of the $\BER$, one needs to compute the unique positive minimizer of $F(\tau)$ in \eqref{eq:F}. The function $F$ is strictly convex, hence this can be done numerically in an efficient way. Due to convexity,  $\tau_*$ can also be described as the unique solution to the first order optimality conditions of the minimization program (see Lemma \ref{lem:tau_prop}). By further analyzing the properties of $\tau_*$, we derive in Section \ref{sec:BPSK_bounds} simple closed-form (upper and lower) bounds on the quantity of interest, namely $Q(1/\tau_*)$.

\vp
\subsubsection{Numerical illustration}  Figure \ref{fig:sim} illustrates the accuracy of the prediction of Theorem \ref{thm:main}. Note that although the theorem requires $n\rightarrow\infty$, the prediction is already accurate for $n$ on the scale of a few hundreds.

\subsection{Simple bounds and high-SNR regime}\label{sec:BPSK_bounds}

We derive simple, closed-form upper and lower bounds on $Q(1/\tau_*)$, the limiting value of the $\BER$. We further show that these bounds are tight. The proof is deferred to Appendix \ref{sec:proof_bounds}.

\begin{thm}[Closed-form bounds]\label{thm:simple}
Let $\tau_*$ be the unique minimizer of \eqref{eq:F}. Then, for all values of $\delta>1/2$ and all values of $\SNR>0$, it holds,
\begin{align}\label{eq:simple}
Q(\sqrt{\delta\cdot\SNR})< Q(1/\tau_*) \leq Q(\sqrt{(\delta-1/2)\cdot\SNR}).
\end{align}
Furthermore, the upper bound becomes tight as $\SNR\rightarrow+\infty$.
\end{thm}

In view of Theorem \ref{thm:main}, the statement in \eqref{eq:simple} directly establishes upper and lower bounds on the (asymptotic) $\BER$ performance of the BRO. These bounds are given in closed-form and are simple functions of $\delta$ and of $\SNR$.

As stated in the theorem, the upper bound in \eqref{eq:simple} becomes tight in the high-SNR regime. Hence,
for $\SNR\gg1$, in the limit of $n\rightarrow\infty$,
\begin{align}\label{eq:BRO_high}
\BER \approx Q(\sqrt{\left(\delta-1/2\right)\cdot\SNR}).
\end{align}
A formal statement of this result is given in Theorem \ref{thm:high-SNR} in Appendix \ref{sec:proof_bounds}. The fact that $\tau_*\approx1/\sqrt{(\delta-1/2)SNR}$ when $\SNR\gg1$, can be intuitively understood as follows: at high-SNR we expect $\tau_*$ to be going to zero (correspondingly $\BER$ or $Q(1/\tau_*)$ to be small). When this is the case, the last two summands in \eqref{eq:F} are negligible; then, $\tau_*$ is the solution to $\min_{\tau>0}{\tau}\left(\delta-\frac{1}{2}\right)+\frac{1/\SNR}{\tau}$, which gives the derided result.

For illustration, in Figure \ref{fig:compare} we have plotted the high-SNR expression for the SER in \eqref{eq:BRO_high} versus its exact value as predicted by Theorem \ref{thm:main}. It is seen that, as already discussed, the high-SNR expression is an upper bound, and in fact a good proxy, for the true probability of error at all values of SNR. The approximation becomes better with increasing $\delta$.

Finally, in Section \ref{sec:BPSK_MFB}, we show that the lower bound $Q(\sqrt{\delta\cdot \SNR})$ has an operational meaning: it is equal to the bit error probability of an isolated bit transmission over the channel, which is also known as the \emph{matched filter bound} in digital communications.


\begin{figure}[t]
    \centering
    \includegraphics[width=0.5\textwidth]{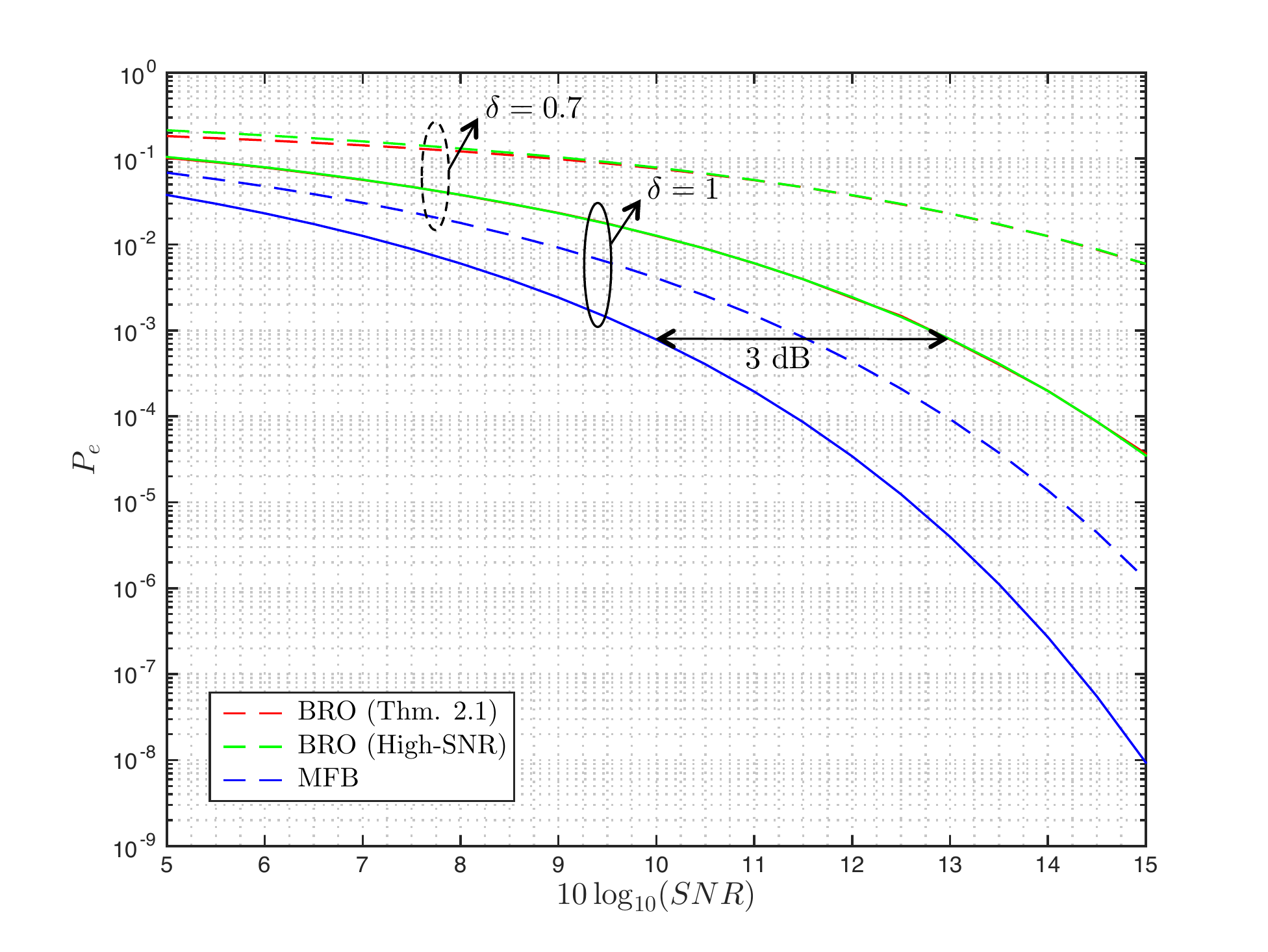}
    \caption{ {Symbol-error probability  of the BRO (in red, see Theorem \ref{thm:main}) in comparison to high-SNR approximation (in green, see \eqref{eq:BRO_high}) and to the matched filter bound (in blue, see \eqref{eq:MFB_high}), for $\delta=0.7$ (dashed lines) and $\delta=1$ (solid lines). Theorem~\ref{thm:simple} successfully predicts that the red curves are sandwiched between the corresponding green (upper-bound) and blue ones (lower-bound).
    }
    }
    \label{fig:compare}
\end{figure}

\subsection{Comparison to the matched filter bound}\label{sec:BPSK_MFB}

Here, we compare the performance of the BRO to an idealistic case, where all $n-1$, but $1$, bits of $\x_0$ are known to us.  As is customary in the field, we refer to the symbol error probability of this case as the \emph{matched filter bound} (MFB) and denote it by $\PeMF$. The MFB corresponds to the probability of error in detecting (say) $\x_{0,n}\in\{\pm1\}$ from:
$
\tilde\y = \x_{0,n}\ab_n + \z,
$
where  $\tilde\y=\y-\sum_{i=1}^{n-1}\x_{0,i}\ab_i$ is assumed known, and, $\ab_i$ denotes the $i^{th}$ column of $\A$. (This can be equivalently thought of as the error probability of an isolated transmission of only the last bit over the channel.)
The ML estimate is equal to the sign of the projection of the vector $\tilde{\y}$ to the direction of $\ab_n$. Without loss of generality we assume that $\x_{0,n}=+1$. Then, the output of the matched filter becomes $\sign(\tilde{X})$, where
$$
\tilde{X} = \|\ab_n\|^2+\sigma^2\tilde{\z}_n,
$$
and $\tilde{\z}_n\sim\Nn(0,1)$. Recall that the entries of the $m$-dimensional vector $\ab_n$ are iid $\Nn(0,1/n)$, so it holds $\plim_{\nlim}\|\ab_n\|=\delta$. Hence,
\begin{align}\label{eq:MFB_high}
\lim_{n\rightarrow\infty}\PeMF = \lim_{n\rightarrow\infty}\Pro(\tilde X<0) = Q(\sqrt{\delta\cdot\SNR}).
\end{align}

First, observe that this formula coincides with the lower bound on the probability of error of the BRO derived in Theorem \ref{thm:simple}. Combined, they establish formally that the MFB is (strictly) better than the BRO. Of course, this is naturally expected since the former is an idealistic scheme.

Next, when compared to the upper-bound on the probability of error of the BRO derived in Theorem \ref{thm:simple}, the formula in \eqref{eq:MFB_high}, leads to the following conclusion:

\begin{center}
\emph{
The BRO achieves a desired symbol-error probability at a higher SNR value by at most $10\log_{10}\frac{\delta}{\delta-1/2}$\dB than that predicted by the MFB.
}
\end{center}

\vp
In particular, in the square case ($\delta=1$), where the number of receive and transmit antennas are the same, the BRO is 3\dB off the MFB (cf., Figure~\ref{fig:compare}). When the number of receive antennas is much larger, i.e., when $\delta\gg 1$, then the performance of the BRO approaches the MFB.

\subsection{Box-relaxation vs Zero-forcing}\label{sec:BPSK_ZF}

In this section, we use Theorem \ref{thm:main} to compare the performance of the BRO to another widely used decoder, namely the \emph{zero-forcing} (ZF) decoder. The ZF decoder  obtains an estimate $\x^*_{\text{ZF}}$
 of $\x_0$ as follows
\begin{subequations}\label{eq:ZF_algo}
\begin{align}
\hat\x_{\text{ZF}} = \arg\min_{\x\in \R^{n}}\|\y-\A\x\|_2,\label{eq:ZF}\\
\x^*_{\text{ZF}} = \sign(\hat\x_\text{ZF}).
\end{align}
\end{subequations}
Observe that this is very similar to the BRO, only that in \eqref{eq:ZF} the minimization is \emph{unconstrained}. Therefore, in contrast to the BRO, for the ZF decoder we require $\delta>1$, i.e., the number of receive antennas be larger than the number of transmit antennas. When this is the case and $n$ is large, $\A$ is full column-rank with probability one, and \eqref{eq:ZF} has a unique closed-form solution:
\begin{align}\label{eq:ZF_sol}
\hat\x_{\text{ZF}} = (\A^T\A)^{-1}\A^T\y.
\end{align}

In particular, it is a well-known result in the literature how to use standard tools from random matrix theory to derive the symbol-error probability of the ZF decoder (e.g. \cite{hassibi2005sphere}). For convenience of the reader, we briefly summarize the main idea here.
Without loss of generality, consider the last bit $\x_n$ of $\x$. Further let $\A=\Qb\Rb$ be the QR decomposition of $\A$, such that $\Qb\in\R^{m\times n}$ is a matrix with orthogonal columns and $\Rb\in\R^{n \times n}$ is upper triangular. Define $\tilde{\y}:=\Qb^T\y$ and $\tilde{\z}:=\Qb^T\z$ and note that
$$\tilde{\y}_n=\Rb_{nn}\x_n+\tilde{\z}_n,$$
where $\Rb_{nn}$ is the $n^\text{th}$ diagonal element of $\Rb$.
From the rotational invariance of the Gaussian distribution, it holds $\tilde{\z}_n\sim\mathcal{N}(0, \sigma^2)$. Next, we use the following well-known facts, e.g., \cite[Lem.~1]{hassibi2005sphere}: (i) $\Qb$ and $\Rb$ are independent matrices. Hence, $\tilde{\z}_n$ is independent of $\Rb_{nn}$; (ii) $\Rb_{nn}$ 	is such that $n\Rb^2_{nn}$ is $\chi^2$ random variable with $(m-n+1)$ degrees of freedom.
Thus, by the corresponding formula for BPSK single-input single-output (SISO) Gaussian channel, the symbol-error probability of the zero-forcing decoder is
$$P_e^{ZF}=E_{\gamma_1,\ldots,\gamma_{m-n+1}} \Big[ Q\Big(\sqrt{\frac{\frac{1}{n}\sum_{i=1}^{m-n+1} \gamma_i^2  }{\sigma^2}   }\Big) \Big],$$
where $\gamma_i$'s are iid standard Gaussians $\mathcal{N}(0,1)$. But, $\plim_{\nlim}\frac{\sum_{i=1}^{m-n+1} \gamma_i^2}{n} =  (\delta-1)$, giving
\begin{align}\label{eq:ZF_P}
\lim_{n\rightarrow\infty}P_e^{ZF} = Q(\sqrt{(\delta-1)\cdot\SNR}).
\end{align}

Comparing this formula to the upper bound on the probability of error of the BRO derived in Theorem \ref{thm:simple}, we formally quantify the superiority of the BRO over the ZF decoder:

\begin{center}
\emph{
The BRO achieves the same performance as the ZF decoder at a lower $\SNR$ value by at least $10\log_{10}\left(\frac{\delta-\frac{1}{2}}{\delta-1}\right)$\dB.
}
\end{center}

\vp
This holds for $\delta>1$. However, Theorem \ref{thm:main} further shows that the BRO can successfully decode even when $\delta<1$, and in particular as low as $1/2$.
%

Above, we derived formula \eqref{eq:ZF_P} using tools from random matrix theory. Alternatively, we can obtain the same result using  the CGMT, and the proof technique is very similar to that of Theorem \ref{thm:main}. The use of random-matrix-theory tools for the analysis of the ZF decoder is in large possible because the minimizer $\hat\x_{\text{ZF}}$ of \eqref{eq:ZF} can be expressed in closed-form as a function of $\A$ and $\z$ (see \eqref{eq:ZF_sol}). On the contrary, this is not the case with the BRO decoder and the use of the CGMT is critical for establishing Theorem \ref{thm:main}.

\section{Extension to M-PAM constellations}\label{sec:MPAM}
%
%


\subsection{Setting}
Each transmit antenna sends a symbol $\x_{0,i}$ that take values $$\x_{0,i}\in\Cc:=\{\pm1,\pm3,\ldots,\pm (M-1)\},$$ for some $M=2^b$ and $b$ a positive integer. When each antenna transmits a single bit, i.e. $b=1$,  then $\x_0\in\{\pm1\}^n$ and the setting is the same as in Section \ref{sec:result}. As always, we assume additive Gaussian noise of variance $\sigma^2$.

The ML decoder is given by $\min_{\x\in\Cc^n}\|\y-\A\x\|_2$, but it is often computationally intractable for large number of receive/transmit antennas. We consider, the natural extension of the box-relaxation decoder for BPSK in \eqref{eq:algo}. Specifically, for M-PAM symbol transmission, the BRO outputs an estimate $\x^*$ of $\x_0$ as follows:
\begin{subequations}\label{eq:algo_MPAM}
\begin{align}
\hat\x = \arg\min_{-(M-1)\leq\x_i\leq (M-1)}\|\y-\A\x\|_2,\label{eq:LASSO_MPAM}\\
\x^*_i = \arg\min_{c\in\Cc}|\hat\x_i - c|. \label{eq:LASSO_thr}
\end{align}
\end{subequations}
The optimization in \eqref{eq:LASSO_MPAM} is convex, and \eqref{eq:LASSO_thr} simply selects the symbol value $c$ that is closest to the solution $\hat\x_i$ among a total of $M$ choices: $\{\pm1,\pm3,\ldots,\pm (M-1)\}$. Therefore, the proposed decoder is computationally efficient.
In the next section, we evaluate its error-rate performance. 

\subsection{SER performance}
Theorem \ref{thm:BER_PAM} below precisely characterizes the large-system limit of the SER of the BRO in \eqref{eq:algo_MPAM} under an M-PAM transmission. We assume that a \emph{typical sequence} of symbols is sent over the channel, i.e., each transmitted symbol $\x_{0,i}$ takes values $\{\pm1,\pm3,\ldots,\pm (M-1)\}$ with equal probability $1/M$. The result extends to other distributions over the constellation, but for simplicity we focus on this typical case. For a typical sequence, the average power of the transmitted vector $\x_0$ is
$$\E[\x_{0,i}^2]=({2}/{M})\sum_{i=1,3\ldots,M-1}i^2={(M^2-1)}/{3}.$$
Therefore,  the $\SNR$ of the system becomes
\begin{align}\label{eq:SNR_M}
\SNR = {{(M^2-1)}}/{3\sigma^2}.
\end{align}


\vp
\begin{thm}[$\SER$ for M-PAM]\label{thm:BER_PAM}
Let $\SER$ denote the symbol error rate of the detection scheme in \eqref{eq:algo_MPAM}, for a typical transmitted signal $\x_0$ such that each symbol $\x_{0,i}$ takes values $\{\pm1,\pm3,\ldots,\pm (M-1)\}$ with equal probability $1/M$. Fix a constant noise variance $\sigma^2$ (eqv., a constant $\SNR$ as in \eqref{eq:SNR_M}) and a constant $\delta\in(1-\frac{1}{M},+\infty)$. Then, in the limit of  $m,n\rightarrow\infty,~{m}/{n}=\delta$, it holds:
$$
\plimn~\SER = 2\left(1-\frac{1}{M}\right)Q\left(\frac{1}{\tau_*}\right),
$$
where $\tau_*$ is the unique positive minimizer of the strictly convex function $\FM:(0,+\infty)\rightarrow\R$ defined as:
\begin{align}\label{eq:FM}
\hspace{-6pt}\FM(\tau):=\frac{\tau}{2}\left(\delta-\frac{M-1}{M}\right)+\frac{\sigma^2}{2\tau} + \frac{1}{M}\sum_{k=2,4,\ldots,2(M-1)} S(\tau; k),
\end{align}
with,
\begin{align}\label{eq:S_thm}
S(\tau;k) := \left(\tau + \frac{k^2}{\tau}\right)Q\left(\frac{k}{\tau}\right) - \frac{k}{\sqrt{2\pi}} e^{-\frac{k^2}{2\tau^2}}.
\end{align}
\end{thm}
%

\begin{figure}[t]
    \centering
    \includegraphics[width=0.45\textwidth]{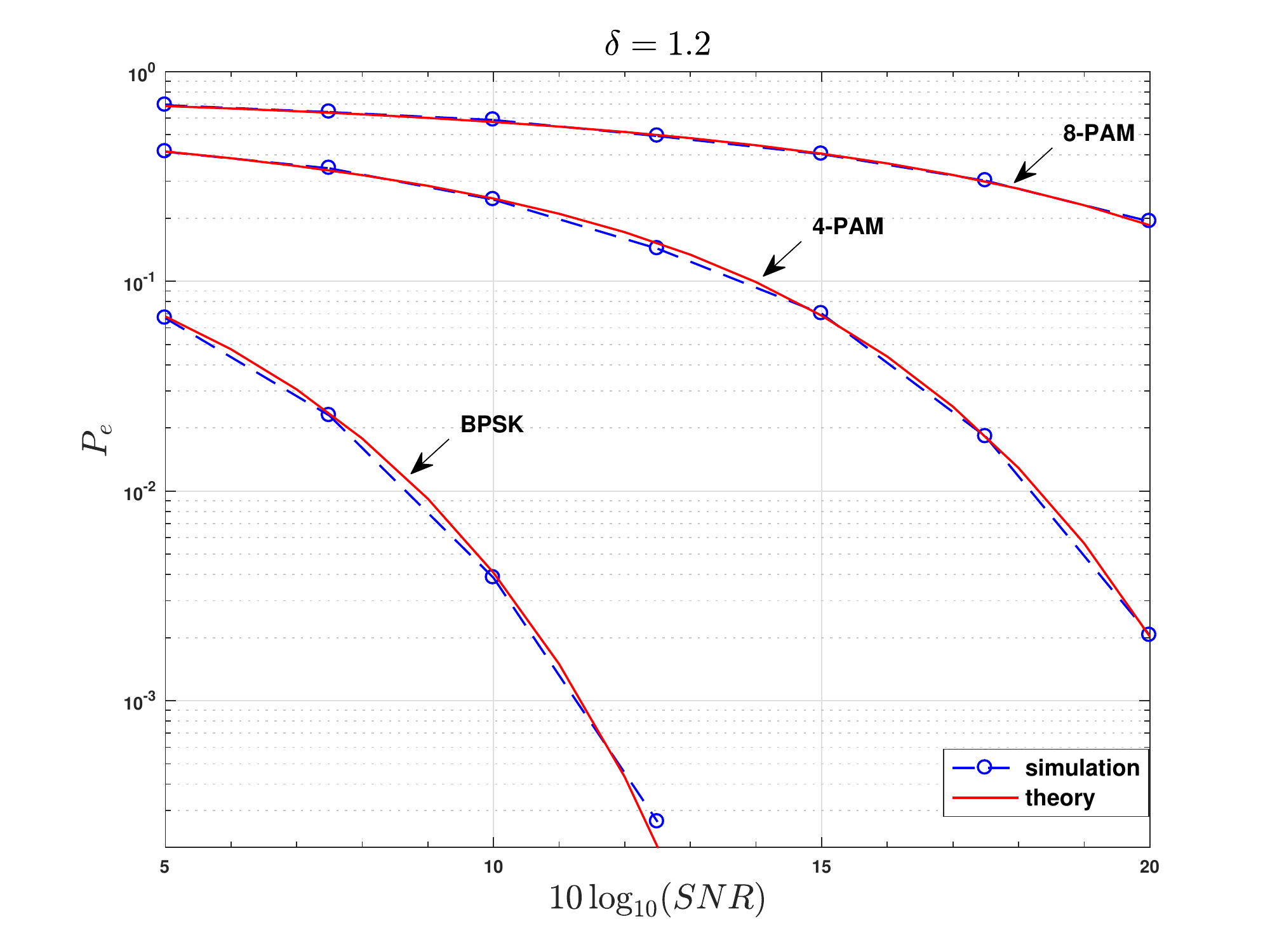}
    \caption{\captionsize{Symbol error probability of the Box Relaxation Optimization (BRO) in \eqref{eq:algo_MPAM}  as a function of the $\SNR$ for BPSK, 4-PAM and 8-PAM signals. The theoretical prediction follows from Theorem \ref{thm:BER_PAM}. For the simulations, we used n = 512 and $\delta=1.2$. The data are averages over 20 independent realizations of the channel matrix and of the noise vector for each value of the SNR.
    }
    }
    \label{fig:BER_MPAM}
\end{figure}

\vp
Theorem \ref{thm:BER_PAM} generalizes Theorem \ref{thm:main}, and the former reproduces the  latter for $M=2$. Figure \ref{fig:BER_MPAM} illustrates the accuracy of the prediction. The proof of the theorem is defered to  Appendix \ref{sec:proof_PAM}. 

Most of the remarks that followed the statement of Theorem \ref{thm:main} in Section \ref{sec:result}, are readily extended to  general M-PAM constellations. The guarantees of Theorem \ref{thm:BER_PAM} hold as long as the ratio of transmit to receive antennas $\delta$ is larger than $1-1/M$. Thus, successful transmission is possible with fewer number of receive than transmit antennas. The minimum allowed ratio increases for higher-order constellations.
Similar to Theorem \ref{thm:simple}, we can show the following simple upper bound on probability of error $\Pe$ for all values of $\SNR$:
\begin{align}\label{eq:up_M}
\lim_{\nlim} \Pe\leq 2\Big(1-\frac{1}{M}\Big)  Q\left({\sqrt{\Big(\delta-1+\frac{1}{M}\Big) \Big(\frac{3}{M^2-1}\Big)\SNR} }\right).
\end{align}
Moreover, the bound is \emph{tight} at high-SNR. Of course, for $M=2$, this coincides with the upper bound in \eqref{eq:simple}.


\section{Proof of main result}\label{sec:proof}
This section includes the proof of Theorem \ref{thm:main}. In fact, towards proving the theorem, we obtain a more general result which is stated as Theorem \ref{thm:general} below.

For simplicity, we make use of the following notation onwards. We say that an event $\Ec$ holds with probability approaching 1 (\emph{w.p.a.1}) if $\lim_{n\rightarrow\infty}\Pro(\Ec)=1$. Also, we use the following shorthands: $X_n\rP X$ to denote convergence in probability;  $X\stackrel{d}{=}Y$ to denote that the random variables $X$ and $Y$ have the same distribution; and, $\|\cdot\|$ to denote the n-dimensional Euclidean norm.

\subsection{Main technical result}
As far as the performance is concerned, we can assume without loss of generality that  $\x_0=+\mathbf{1}_n=(1,1,\ldots,1)$.
Also, it is convenient to re-write  \eqref{eq:LASSO} by changing the variable to the \emph{error vector} $\w:=\x-\x_0=\x-\mathbf{1}$:
\begin{align}\label{eq:LASSOw}
\hat\w := \arg\min_{-2\leq\w_i\leq 0} \|\z-\A\w\|.
\end{align}
\vspace{-2pt}
Then, observe that the SER defined in \eqref{eq:BER} is written in terms of the error vector $\w$ as:
\begin{align}\label{eq:BERw}
\BER=\frac{1}{n}\sum_{i=1}^n\ind{\hat\w_i\leq -1}.
\end{align}

\begin{figure}[t]
    \centering
    \includegraphics[width=1\columnwidth]{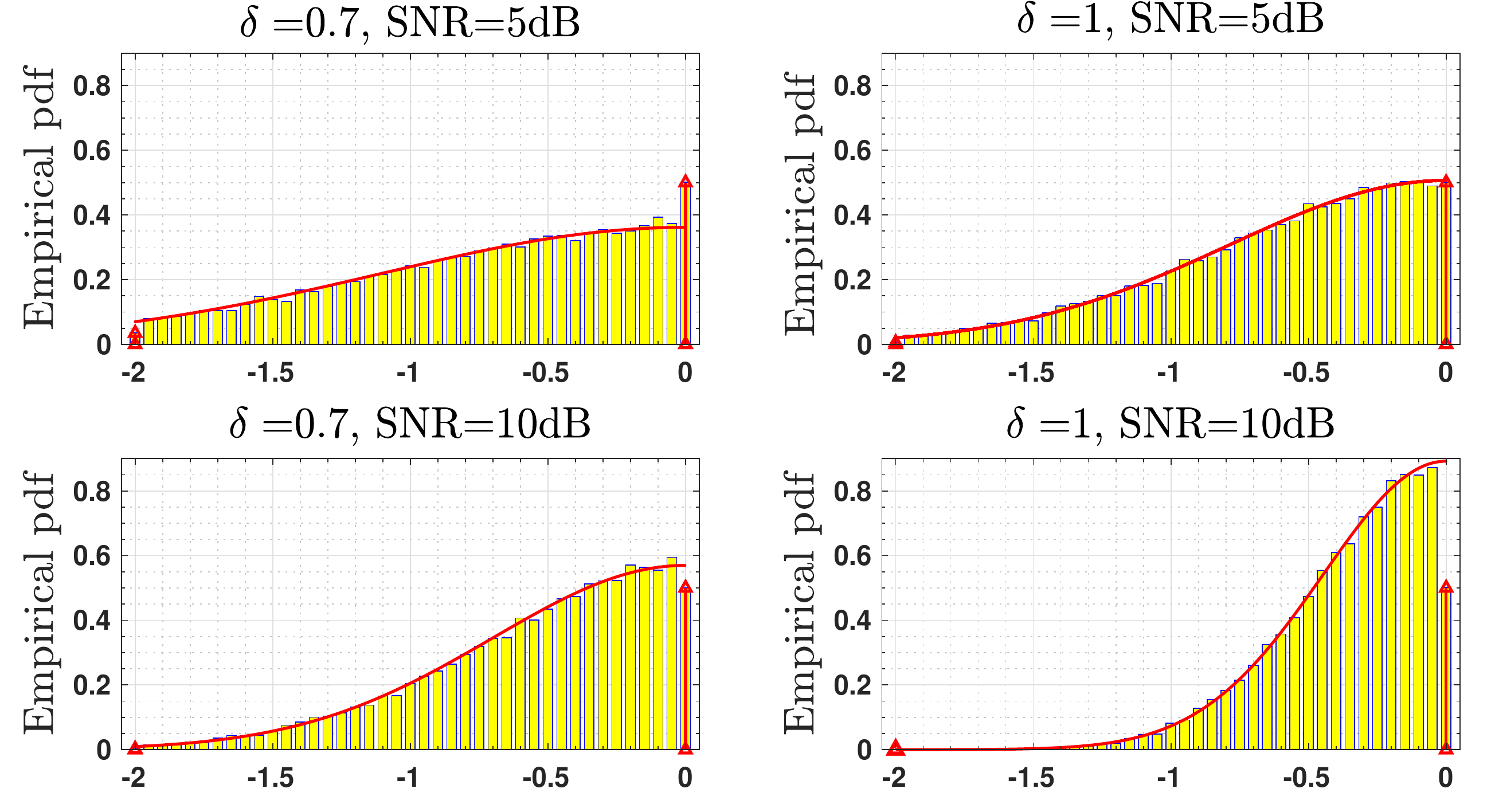}
    \caption{\captionsize{Empirical distribution of the error vector $\w:=\hat\x-\x_0$ (conditioned on $\x_0=+\mathbf{1}$) for the solution $\hat\x$ of the BRO. The empirical histograms shown are averages over $200$ realizations of the channel matrix and of the noise vector for $n=256$ number of transmit antennas.  They are compared to the asymptotic limiting distribution predicted by Theorem \ref{thm:general}, see \eqref{eq:W}. The limiting density is supported in the interval $[-2,0]$ and has point masses both at $-2$ and $0$. Different values of $\delta$ and of $\SNR$ are shown.
    }
    }
    \label{fig:Prob}
\end{figure}

The following theorem characterizes the limit of the empirical distribution of the optimal solution $\hat\w$ in \eqref{eq:LASSOw}, and yields Theorem \ref{thm:main} as a corollary.

\begin{thm}[Lipschitz metrics and empirical distribution]\label{thm:general}
Recall the definition of $\tau_*$ in Theorem \ref{thm:main}, and assume, without loss of generality, that $\x_0=+\mathbf{1}$.  Let $\wh$ be as in \eqref{eq:LASSOw} and consider its (normalized) empirical density function $$\mu_{\wh} := n^{-1}\sum_{i=1}^{n}\delta_{\wh_i}.$$ Further, consider the function $\theta:\R\rightarrow[-2,0]$:
\begin{align}\nn
\theta(\gamma) := \begin{cases}
0 &,\text{if } \gamma\geq 0,\\
\tau_*\gamma&,\text{if } -\frac{2}{\tau_*}\leq \gamma < 0, \\
-2&,\text{if }  \gamma< -\frac{2}{\tau_*},
\end{cases}
\end{align}
and let $\mu_W$ be the density measure of a random variable $W$
\begin{align}\label{eq:W}
W \stackrel{d}{=} \theta(\Nn(0,1)).
\end{align}
The following are true:

\vp
\noindent{{(a).}}~~ $\mu_{\wt}$ converges weakly in probability to $\mu_{W}$. \footnote{\label{foot:weak}Note that $\mu_{\wt}$ defines a (sequence of) \emph{random} probability measure(s); on the other hand, $\mu_W$ is a deterministic measure.  We use terminology that is standard in the theory of random matrices and say that a sequence of random measures $\mu_n$ converges weakly to a deterministic measure $\mu$ if for every continuous compactly supported $\psi$: $\int \psi \mathrm{d}\mu_n \rP \int \psi \mathrm{d}\mu$ (see for example \cite[pg.~160]{tao2012topics}).}

\vp
\noindent{{(b).}}~~ For all Lipschitz functions $\psi:\R\rightarrow\R$ with Lipschitz constant $L$ (independent of $n$), it holds $$\frac{1}{n}\sum_{i=1}^{n}\psi(\wh_i)\rP \E_{W}[\psi(W)].$$
%
%

\end{thm}

Theorem \ref{thm:general} is the main technical result of this paper. In Section \ref{sec:Proof_main} we show how it can be used to prove Theorem \ref{thm:main}. Next, in Section \ref{sec:indep} we rely again on Theorem \ref{thm:general} to prove that error events for any fixed number of bits are asymptotically independent.
The rest of Section \ref{sec:proof} is devoted to the proof of Theorem \ref{thm:general}.
%


\subsection{Proof of Theorem \ref{thm:main}}\label{sec:Proof_main}
On the one hand, by \eqref{eq:BERw}, it suffices to prove that
$\frac{1}{n}\sum_{i=1}^n\ind{\wh_i\leq -1} \rP Q(1/\tau_*).$
On the other hand, it is easily checked that $\E_W[\ind{W\leq -1}] = \E_{\gamma\sim\Nn(0,1)}[\ind{\gamma\leq -1/\tau_*}]=Q(1/\tau_*).$ Note that the indicator function $\ind{W\leq -1}$ is not Lipschitz, so we cannot directly apply Theorem \ref{thm:general}(b). However, since the discontinuity point (i.e., $-1$) of the indicator function has $\mu_W$-measure zero, and also $W$ is a \emph{continuous} random variable, one can appropriately approximate the indicator with Lipschitz functions and conclude the desired based on Theorem \ref{thm:general}(b).
%
 This is a somewhat standard argument, but we reproduce a detailed proof of the claim in Lemma \ref{lem:ind} in Appendix \ref{sec:proof_proof} for completeness.


\subsection{Independence of Error Events}\label{sec:indep}

Here, we obtain as a corollary of Theorem \ref{thm:general} that error events for any fixed number of bits are \emph{asymptotically} independent. We defer the proof of the corollary to Appendix \ref{sec:indep_proof}.

\begin{cor}[Independence of error events]\label{cor:indep}
 Under the notation and definition of Theorem \ref{thm:general}, let $\psi_i:\R\rightarrow\R, i=1,\ldots,k$ be bounded Lipschitz functions for fixed $k\geq 2$. Then, it holds
$$
n^{-k}\sum_{1\leq i_1,\ldots,i_k\leq n} \psi_1(\hat\w_{i_1})\cdots\psi_k(\hat\w_{i_k}) \rP \prod_{\ell=1}^{k}\E[\psi_\ell(W_\ell)],
$$
where the expectations of the right-hand side are with respect to $W_1,\ldots,W_k$ that are iid random variables distributed as $\theta(\Nn(0,1))$. Moreover, it holds
$$
n^{-k}\sum_{1\leq i_1,\ldots,i_k\leq n} \ind{\hat\w_{i_1}\leq-1,\ldots,\hat\w_{i_k}\leq-1} \rP (Q(1/\tau_*))^k.
$$
\end{cor}
\subsection{The convex Gaussian min-max theorem}
The fundamental tool behind our analysis is the convex Gaussian min-max theorem (CGMT) \citep{Master,COLT15}. The CGMT associates with a primary optimization (PO) problem a simplified auxiliary optimization (AO) problem from which we can tightly infer properties of the original (PO), such as the
optimal cost, the optimal solution, etc.. In particular, the (PO) and (AO) problems are defined respectively as follows:
\begin{subequations}\label{eq:POAO}
\begin{align}
\label{eq:PO_gen}
\Phi(\G)&:= \min_{\w\in\Sc_\w}~\max_{\ub\in\Sc_\ub}~ \ub^T\G\w + \psiubw,\\
\label{eq:AO_gen}
\phi(\g,\h)&:= \min_{\w\in\Sc_\w}~\max_{\ub\in\Sc_\ub}~ \|\w\|\g^T\ub - \|\ub\|\h^T\w + \psiubw,
\end{align}
\end{subequations}
where $\G\in\R^{m\times n}, \g\in\R^m, \h\in\R^n$, $\Sc_\w\subset\R^n,\Sc_\ub\subset\R^m$ and $\psi:\R^n\times\R^m\rightarrow\R$. Denote $\w_\Phi:=\w_\Phi(\G)$ and $\w_\phi:=\w_\phi(\g,\h)$ any optimal minimizers in \eqref{eq:PO_gen} and \eqref{eq:AO_gen}, respectively.
%
Further let $\Sc_\w,\Sc_\ub$ be convex and compact sets, $\psi$ be continuous and convex-concave on $\Sc_\w\times\Sc_\ub$, and, $\G,\g$ and $\h$ all have entries iid standard normal.

\begin{thm}[CGMT,~\citep{Master}]\label{thm:CGMT}
Let $\Sc$ be an arbitrary open subset of $\Sc_\w$ and $\Scc=\Sc_\w/\Sc$. Denote $\phi_\Scc(\g,\h)$ the optimal cost of the optimization in \eqref{eq:AO_gen}, when the minimization over $\w$ is now constrained over $\w\in\Scc$.
 Suppose there exist constants $\phio$ and $\eta>0$ such that in the limit of $n\rightarrow+\infty$
 it holds w.p.a.1: (i) $\phi(\g,\h)\leq\phio + \eta$, and, (ii) $\phi_{\Scc}(\g,\h)\geq \phio + 2\eta$.
%
  Then, $\lim_{n\rightarrow\infty}\Pr(\w_\Phi \in \Sc)=1$.
\end{thm}

It is not hard to argue that the conditions (i) and (ii) regarding the optimal cost of the (AO) imply the following for its solution: $\w_\phi \in \Sc$ w.p.a.1. The non-trivial and powerful part of the theorem is that the same conclusion is true for the optimal solution $\w_\Phi$ of the (PO) as well. The CGMT
builds upon a classical result due to Gordon \citep{GorThm}. Gordon's original result is classically used to establish non-asymptotic probabilistic lower bounds on the minimum singular value of Gaussian matrices \citep{vershynin2010introduction}, and has a number of other applications in high-dimensional convex geometry \citep{GorLem,ledoux}.
The idea of combining the GMT with convexity is attributed to Stojnic \citep{StoLASSO}. Thrampoulidis et. al. built and significantly extended on this idea arriving at the CGMT as it appears in \citep[Thm.~6.1]{Master}.

\subsection{Proof of Theorem \ref{thm:general}}

\subsubsection{Strategy}
We will first prove Theorem \ref{thm:general}(b); Part (a) will then follow by standard arguments from the theory of weak convergence.

As mentioned the proof is based on the use of the CGMT. The first step is to identify the (PO) and the (AO), such that $\wh$ is optimal for the (PO). Then, our goal is to apply Theorem \ref{thm:CGMT} to the following set
\begin{align}\label{eq:S_eps}
\Sc_\eps := \{\vb~:~ \big|n^{-1}\sum_{i=1}^n\psi({\vb_i}) - \E_W[\psi(W)] \big|< \eps \},
\end{align}
where $\eps>0$ is arbitrary. To see that this is desired note that if for all $\eps>0$ it holds $\w\in\Sc_\eps$ w.p.a.1, then $n^{-1}\sum_{i=1}^n\psi({\w_i}) \rP \E_W[\psi(W)]$. Thus, the bulk of the proof amounts to checking that the conditions of Theorem \ref{thm:CGMT} are satisfied for $\Sc_\eps$ in \eqref{eq:S_eps}. For the rest of the proof, we fix $\eps>0$ and denote $\Sc:=\Sc_\eps$, for convenience

%
%


\vp
\subsubsection{Identifying the (PO) and the (AO)}
Using the CGMT for the analysis of the  $\BER$, requires as a first step  expressing the optimization in \eqref{eq:LASSO} in the form of a (PO) as it appears in \eqref{eq:PO_gen}. It is easy to see that \eqref{eq:LASSOw} is equivalent to 
\begin{align}
\frac{1}{\sqrt{n}}~\min_{-2\leq\w_i\leq 0}\max_{\|\ub\|\leq 1} \ub^T\A\w-\ub^T\z.
\end{align}
Observe that the constraint sets above are both convex and compact; also, the objective function is convex in $\w$ and concave in $\ub$. These are consistent with the requirements of the CGMT. The corresponding (AO) problem becomes:
\vspace{-5pt}
\begin{align}\label{eq:AOw}
\phi(\g,\h):=\frac{1}{{n}}\min_{-2\leq\w_i\leq 0}\max_{\|\ub\|\leq 1} (\|\w\|\g-\sqrt{n}\z)^T\ub-\|\ub\|\h^T\w.
\end{align}
Note the normalization to account for the variance of the entries of $\A$. Onwards, we refer to the optimization in \eqref{eq:AOw} as the (AO) problem.
\vp
\subsubsection{Simplifying the (AO)}\label{sec:proof_simplify}
  We begin by simplifying the (AO) problem as it appears in \eqref{eq:AOw}. First, since both $\g$ and $\z$ have entries iid Gaussian, then, the vector $\|\w\|\g-\sqrt{n}\z$ has entries iid $\Nn(0,\sqrt{\|\w\|^2+n\sigma^2})$. Hence, for our purposes and using some abuse of notation so that $\g$ continues to denote a vector with iid standard normal entries, the first term in \eqref{eq:AOw} can be treated as $\sqrt{\|\w\|^2+n\sigma^2}\g^T\ub$, instead. As a next step,
fix the norm of $\ub$ to say $\|\ub\|=\beta$. Optimizing over its direction is now straightforward, and gives
$$\min_{-2\leq\w_i\leq 0}~\max_{0\leq \beta\leq 1} \frac{\beta}{{n}}\left(\sqrt{\|\w\|^2+n\sigma^2}\|\g\|-\h^T\w\right).$$
In fact, it is easy to now further optimize over $\beta$ as well: its optimizing value is $1$ if the term in the parenthesis is non-negative, and, is 0 otherwise. With this, the (AO) simplifies to the following:
\begin{align}\label{eq:AOw11}
\phi(\g,\h) = \min_{-2\leq\w_i\leq 0}\big(\sqrt{\frac{\|\w\|^2}{n}+\sigma^2}\frac{\|\g\|}{\sqrt{n}}-\frac{1}{{n}}\h^T\w\big)_+,
\end{align}
where we used the notation $(\cdot)_+:=\max\{\cdot,0\}$.

In order to perform the optimization over $\w$, we will express the ``square-root term" $\chi:=\chi(\w):=\sqrt{\|\w\|^2/n+\sigma^2}$  in a variational form. First, observe that  all  $\w\in[-2,0]^n$ satisfy $\sigma^2\leq\chi\leq 4+\sigma^2:=T$. Hence, we can write
$${\chi}=\min_{0\leq\tau\leq T}\frac{\tau}{2} + \frac{\chi^2}{2\tau}.$$
With this trick, the minimization over the entries of $\w$ becomes separable as follows:
\begin{align}\label{eq:AOw25}
\min_{\substack{0\leq\tau\leq T}}\frac{\tau\|\g\|}{2\sqrt{n}} + \frac{\sigg\|\g\|}{2\tau\sqrt{n}}+ \frac{1}{n}\sum_{i=1}^n \min_{-2\leq\w_i\leq 0} \left\{\frac{\gm}{2\tau \sqrt{n}}\w_i^2-{\h_i}\w_i\right\}.
\end{align}
In particular,
the optimal $\tilde\w_i:=\tilde\w_i(\g,\h)$ of \eqref{eq:AOw} satisfies
\begin{align}\label{eq:w_opt}
\tilde\w_i = \begin{cases}
0&,\text{if } \h_i\geq 0,\\
\frac{\taut\sqrt{n}}{\gm}\h_i&,\text{if } -\frac{2\gm}{\taut\sqrt{n}}\leq \h_i < 0, \\
-2&,\text{if }  \h_i< -\frac{2\gm}{\taut\sqrt{n}}.
\end{cases}
\end{align}
where, $\taut:=\taut(\g,\h)$ is the solution to the following:
%
\begin{align}\label{eq:AOw3}
\phi(\g,\h)=\Big(\min_{0\leq\tau\leq T} \frac{\tau\|\g\|}{2\sqrt{n}} + \frac{\sigg\gm}{2\tau\sqrt{n}} + \frac{1}{{n}}\sum_{i=1}^n \upsilon_n\left(\frac{\tau\sqrt{n}}{\|\g\|};\h_i\right) \Big)_+,
\end{align}
\vspace{-5pt}
with, $\upsilon_n:$
\begin{align}
\upsilon_n(\alpha;h) :=\begin{cases} 0&, \text{if } h\geq 0,\\
-\frac{\alpha}{2}h^2&, \text{if } -\frac{2}{\alpha}\leq h < 0, \\
\frac{2}{\alpha} + 2 h&, \text{if } h\leq -\frac{2}{\alpha},
\end{cases}\nn
\end{align}
for all $\alpha>0$ and $h\in\R$.
We remark that the minimization in \eqref{eq:AOw3} is convex. (The easiest way to see this is noting that the objective function in \eqref{eq:AOw25} is jointly convex in $\w$ and $\tau$).

\vp
\subsubsection{Convergence properties of the (AO)}\label{sec:proof_4}
Now that the (AO) is simplified as in \eqref{eq:AOw3}, we study here its behavior in the limit of $m,n\rightarrow\infty$ with $m/n=\delta$.

First, we compute the point-wise (in $\tau$) limit of the objective function in \eqref{eq:AOw3}.
Clearly,
 \begin{align}\label{eq:delta_lim}
 {\gm}/{\sqrt{n}}\rP \sqrt\delta.
 \end{align}
 Also, conditioned on the value of $n^{-1/2}\|\g\|$, the random variable $\sum_{i=1}^n\upsilon_n({\tau\sqrt{n}}/{\|\g\|};\h_i)$ is a sum of absolutely integrable iid random variables. Hence, combining the WLLN with \eqref{eq:delta_lim} it follows that, for all $\tau>0$,
 $$
\frac{1}{{n}}\sum_{i=1}^n \upsilon_n\left(\frac{\tau\sqrt{n}}{\|\g\|};\h_i\right)\rP
Y\left(\frac{\tau}{\sqrt{\delta}}\right)
$$
where,
\begin{align}
Y(\alpha) &:= -\frac{\alpha}{2}\int_{0}^{\frac{2}{\alpha}}h^2p(h)\mathrm{d}h + \frac{2}{\alpha}Q\left(\frac{2}{\alpha}\right) - 2\int_{\frac{2}{\alpha}}^{\infty}hp(h)\mathrm{d}h\nn
\\
&=-\frac{\alpha}{4} + \frac{\alpha}{2}\int_{\frac{2}{\alpha}}^{\infty}\left(h-\frac{2}{\alpha}\right)^2p(h)\mathrm{d}h.\label{eq:Y}
\end{align}

Next, the point-wise convergence implies uniform convergence, thanks to convexity. This follows from \uniform, which is also known in the literature of estimation theory as the convexity lemma: point wise convergence of convex functions implies uniform convergence in
compact subsets (see also \cite[Lem.~7.75]{chakraborty2008statistical}).
 Hence, the random optimization in \eqref{eq:AOw3} converges to the following deterministic optimization (for convenience we rescale the optimization variable $\tau$ as follows: $\tau:=\frac{\tau}{\sqrt{\delta}}$):
 \begin{align}\label{eq:lim_phi}
\overline{\phi}:= \min_{0\leq \tau\leq (T/\sqrt{\delta})}\frac{\tau\delta}{2}+\frac{\sigg}{2\tau}+Y(\tau).
 \end{align}
 Expanding the square in the second summand in \eqref{eq:Y} and applying integration by parts, it can be checked that the objective function in \eqref{eq:lim_phi} is exactly $2F(\tau)$, where $F(\tau)$ is defined in \eqref{eq:F}.

When $\delta>1/2$, all summands in the objective function in \eqref{eq:lim_phi} are  non-negative for all $\tau>0$. Thus, $\overline{\phi}\geq 0$, and consequently (recall \eqref{eq:AOw3}),
  \begin{align}\label{eq:phi_lim}
  \phi(\g,\h) \rP \overline{\phi}.
  \end{align}

We remark that the objective function in \eqref{eq:lim_phi} is strictly convex in the optimization variable $\tau$. (Its convexity follows directly as it is the point-wise limit of convex functions in \eqref{eq:AOw3}, which is known to be convex. Alternatively, and to further check strict convexity, it can be shown that the second derivative is positive.) Hence, there is a unique minimizer, call it  $\tau_*$.
With these, it only takes a standard argument (e.g., see \cite[Thm.~2.1]{NF36}) to further conclude that the minimizer $\taut(\g,\h)$ of \eqref{eq:AOw3} converges in probability to $\tau_*\sqrt{\delta}$, i.e.
\begin{align}\label{eq:tau_lim}
\delta^{-1/2}\taut(\g,\h) \rP \tau_*.
\end{align}
%


%
%
\vp
\subsubsection{The optimal solution of the (AO)}
We now have all the tools necessary to study the properties of the optimal solution $\wt$ of the (AO). The lemma below establishes that for Lipschitz functions, $\wt\in\Sc$ (recall the definition of $\Sc$ in \eqref{eq:S_eps}).


\begin{lem}[Lipschitz convergence of the (AO)]\label{lem:Lip_AO}
Let $\psi:\R\rightarrow\R$ be $L$-Lipschitz, $\wt=\wt(\g,\h)$ as in \eqref{eq:w_opt}, and random variable $W$ as in the statement of Theorem \ref{thm:general}. It holds,
$\frac{1}{n}\sum_{i=1}^{n}\psi(\wt_i)\rP \E_{W}[\psi(W)].$
%
\end{lem}
\begin{proof}
For $i=1,\ldots,n$, define
$\vb_i:=\theta(\h_i)$ (recall the definition of $\theta$ in the statement of Theorem \ref{thm:general}).
The WLLN gives 
\begin{align}
n^{-1}\sum_{i=1}^{n}\psi(\vb_i)\rP \E_{\gamma\sim\Nn(0,1)}[\psi(\theta(\gamma))] =  \E_{W}[\psi(W)],
\end{align}
where we also used the Gaussianity of $\h_i$ and \eqref{eq:W}.
Hence, it will sufficec for the proof to show that $|n^{-1}\sum_{i=1}^{n}(\psi(\wt_i)-\psi(\vb_i))|\rP 0$.
 We show this using the Lipschitz assumption and \eqref{eq:tau_lim}. First, by the Lipschitz property: 
 \begin{align}\label{eq:L2}
 |\psi(\wt_i)-\psi(\vb_i)|\leq L|\wt_i-\vb_i|.
 \end{align}
 Next, the expression of $\wt$ in \eqref{eq:w_opt}, along with \eqref{eq:delta_lim} and with \eqref{eq:tau_lim}, they can be used to show that the RHS in \eqref{eq:L2} is appropriately small. Formally, writing $\ksi:=\ksi(\g,\h)=\frac{\taut\sqrt{n}}{\gm}$ for simplicity, it follows from the continuous mapping theorem that for some $\eta>0$ (the value of which to be chosen later) we have w.p.a.1: $|\ksi-\tau_*|\leq \eta$, and, $|\frac{2}{\ksi}-\frac{2}{\tau_*}|\leq \eta$. Hence, w.p.a.1:
\begin{align*}|\wt_i-\vb_i|&\leq\max\Big\{
|\tau_*-\ksi||\h_i|\ind{\h_i\geq\max{\{-2/\tau_*,-2/\ksi\}}},\\
&\qquad\qquad
|\tau_*\h_i+2|\ind{-2/\tau_*\leq\h_i\leq-2/\ksi},\\
&\qquad\qquad
|\ksi\h_i+2|\ind{-2/\ksi\leq\h_i\leq-2/\tau_*}
 \Big\}\ \\
 &\leq\eta(\eta+{2}/{\tau_*}) + \eta + \eta(\eta+\tau_*).
 \end{align*}
For any $\zeta>0$, choose $\eta=\min\{\frac{\sqrt{\zeta}}{2}~,~\frac{\zeta}{4}(\frac{1}{\tau_*}+\frac{1+\tau_*}{2})\}$, such that in view of \eqref{eq:L2} $ |\psi(\wt_i)-\psi(\vb_i)|\leq L\zeta$, which completes the proof.
\end{proof}


%
\vp\subsubsection{Satisfying the conditions of the CGMT}
The following result uses Lemma \ref{lem:Lip_AO} and strong-convexity of the (AO) to show that the optimal cost of the (AO) strictly increases when the optimization is constrained outside the set $\Sc$ defined in  \eqref{eq:S_eps}. The proof is deferred to Appendix \ref{sec:proof_strong_AO}.

\begin{lem}[Strong convexity of the (AO)]\label{lem:strong_AO}
Let $\psi:\R\rightarrow\R$ be L-Lipschitz, $W$ a random variable as in the statement of Theorem \ref{thm:general}, and $\Sc:=\Sc_\eps$ the set defined in \eqref{eq:S_eps}. Finally, denote $f(\w):=f(\w;\g,\h)$ the objective function in \eqref{eq:AOw11}.
There exists constant $C>0$, such that the following statement holds w.p.a.1,
$$
\min_{\substack{\w\in[-2,0]^n\\\w\in\Scc}}f(\w;\g,\h) \geq \phi(\g,\h) + \frac{C\eps}{L}.
$$
\end{lem}
%

The lemma above essentially verifies conditions (i) and (ii) of the CGMT Theorem \ref{thm:CGMT}.
To be specific,
let $C$ as in the statement of Lemma \ref{lem:strong_AO}, $\phio$ as in \eqref{eq:lim_phi}, and, choose $\eta := \frac{C\eps}{3L}$. From \eqref{eq:phi_lim} it holds w.p.a.1: $|\phi(\g,\h)-\phio|\leq \eta$. Combine this  with Lemma \ref{lem:strong_AO} to conclude that $\phi_{\Scc}(\g,\h)\geq\phio+2\eta$ w.p.a.1, as desired.


%
\vp
\subsubsection{Completing the proof}
At the end of last section we showed that the conditions of the CGMT Theorem \ref{thm:CGMT} are satisfied. Hence, its application yields that any minimizer $\wh$ of the (PO) in \eqref{eq:LASSOw} satisfies $\wh\in\Sc_\eps$ w.p.a.1. This proves part (b) of Theorem \ref{thm:general}. It remains to prove Part (a). Recall the note in Footnote \ref{foot:weak}: it suffices to prove that
\begin{align}\label{eq:Cc}
\frac{1}{n}\sum_{i=1}^{n}\psi(\wh_i)\rP \E_{W}[\psi(W)],
\end{align} for all continuous functions with compact support. Of course, the statement in \eqref{eq:Cc} is true for Lipschitz continuous functions from part (b) of the theorem. But, continuous compactly supported functions are also bounded. The implication from Lipschitz bounded functions to continuous bounded functions is standard and is part of what is known in the literature as the Portmanteau Theorem; see for example \cite[Thm.~13.16]{klenke2013probability}. 
%



\section{Discussion and Future work}\label{sec:conc}
In this paper we have used the recently developed CGMT framework in \citep{COLT15,Master} to precisely compute the large-system error-rate performance of the popular box-relaxation optimization method for recovering signals from M-ary constellations, when the channel matrix and additive noise are both iid real Gaussians. The derived formulas were previously unknown. Also, the CGMT was previously only used to analyze squared-error performance; here, we illustrate for the first time its use to analyze the error-rate performance of convex optimization-based massive MIMO decoders. 

In future work, we seek to extend the analysis to complex Gaussian channels with symbols originating from complex-valued constellations. At its core, this task requires extending the CGMT to complex-valued Gaussian matrices, an extension that is currently unavailable; thus, it poses a challenging, yet practically important, research direction. What appears more accessible is establishing the universality of our results for iid channels beyond Gaussians.  We believe that this is possible by combining the ideas of our paper for extended use of the CGMT with the techniques in \cite{oymak2015universality}, where the universality property has been proven for the squared-error (rather than for the symbol-error-rate).


%
%

For BPSK signal recovery using the BRO, we proved in Corollary \ref{cor:indep} that error events for any fixed number of bits in the solution of the BRO are iid. This fact has potentially significant consequences to be explored. For example, it implies that, when a block of data is in error, only a few of its bits are. This means that the output of the BRO can be used by various local methods to further reduce the SER. We are planning to explore such implications further in future work.

{\small
\bibliography{compbib}
}
\appendix
\subsection{Supplementary proofs for Section \ref{sec:result}}


\subsubsection{Corollary \ref{cor:main}}\label{sec:cor_main_proof}
The corollary follows from Theorem \ref{thm:main} when combined with the following statement, which we prove here:{``If  $\text{SER}(\A,\z) \rP c   $, for some deterministic constant $c$,  then, $\Pe \rightarrow c. $ ''}


For convenience, let us define the random variable $X:=X(\A,\z):=\text{SER}(\A,\z)$. With this notation, $\Pe = \E_{\A,\z}[X]$. Thus, for any $\epsilon>0$, 
\begin{align*}
P_e  &\leq \E\left[X~|~ |X-c|\leq \epsilon \right] + 
\\&\qquad\quad\E\left[X~|~ |X-c|> \epsilon\right]\cdot\Pro\left(|X-c|> \epsilon\right).
\\&\leq (c+\epsilon)+\Pro(|X-c|> \epsilon),
\end{align*}
where in the second inequality we used the fact that $X\leq 1$. Notice that $(c+\epsilon)+P(|X-c|> \epsilon) \rightarrow (c+\epsilon)$ as $n \rightarrow \infty$, since $X \rP c  $, by assumption. In a similar vein,
\begin{align*}
P_e  &\geq \E\left[X~|~ |X-c|\leq \epsilon \right]\cdot\Pro\left(|X-c|\leq \epsilon\right)
\\&\geq (c-\epsilon)\cdot\Pro(|X-c|> \epsilon),
\end{align*}
%
where, again, $(c-\epsilon)\Pro(|X-c|\leq \epsilon) \rightarrow (c-\epsilon)$ as $n \rightarrow \infty$, since $X \rP c  $.
Since the above hold for all $\eps$, we have shown that $\Pe\rightarrow c$, as desired.

\subsubsection{Proof of Theorem \ref{thm:simple}}\label{sec:proof_bounds}
Here, we prove the first part of the theorem, namely the lower and upper bounds on $Q(1/\tau_*)$. The tightness of the upper bound at high-SNR is shown later in Section \ref{sec:high-SNR}. 
Due to the decreasing nature of the function $Q$, it suffices to prove that
\begin{align}\label{eq:bounds_tau}
\sqrt{(\delta-1/2)\cdot\SNR}<\tau_*^{-1}<\sqrt{\delta\cdot\SNR}.
\end{align}
This is shown in Lemma \ref{lem:tau_prop}(b) below. 
The proof of the lemma builds on understanding the behavior of the function $F(\tau)$ in \eqref{eq:F}. The function $F$ is composed of 4 additive terms. The first is linear in $\tau$ and the second is simply $1/\tau$. We view the remaining terms as a single function of $\tau$, namely $S_2(\tau):=\Big(\tau+\frac{4}{\tau}\Big)Q\left(\frac{2}{\tau}\right)-\sqrt{\frac{2}{\pi}}e^{-\frac{2}{\tau^2}}$, and we gather its properties in Lemma \ref{lem:technical} below. 

\begin{lem}\label{lem:technical}
Fix a positive integer $\ell>0$ and consider the function $S:(0,\infty)\rightarrow\R$ defined as follows:
\begin{align}
S_\ell(\alpha):=S(\alpha;\ell) :=\big(\alpha + \frac{\ell^2}{\alpha}\big)Q\Big(\frac{\ell}{\alpha}\Big) - \frac{1}{\sqrt{2\pi}}\ell e^{-\frac{\ell^2}{2\alpha^2}}.
\end{align}
The following statements are true.

\begin{enumerate}[(a)]
\item The first two derivatives $S_\ell'(\alpha)$ and $S_\ell''(\alpha)$ are given as follows
\begin{align}
S_\ell'(\alpha)&= \frac{\ell}{\alpha\sqrt{2\pi}}e^{-\frac{\ell^2}{2\alpha^2}} + \Big( 1 - \frac{\ell^2}{\alpha^2} \Big)Q\Big(\frac{\ell}{\alpha}\Big).\label{eq:S'} \\
S_\ell''(\alpha)&=2\frac{\ell^2}{\alpha^2}Q\Big(\frac{\ell}{\alpha}\Big). \nn
\end{align}
\item The function $S_\ell(\alpha)$ is strictly convex.
\item The derivative $S_\ell'(\alpha)$ is strictly increasing. Moreover
$$
\lim_{\alpha\rightarrow 0^+}S_\ell(\alpha) = 0 < S_\ell'(\alpha) < \frac{1}{2}= \lim_{\alpha\rightarrow+\infty}S_\ell(\alpha).
$$
\end{enumerate}

\end{lem}

\begin{proof}
Statement (a) follows easily by direct calculations. It can be readily observed that $S_\ell''(\alpha)$ is strictly greater than 0 for all $\alpha>0$. This proves statement (b). For the last statement, we argue as follows: $S_\ell'(\alpha)$ is strictly increasing by strict convexity of $S_\ell(\alpha)$. Thus, it suffices to compute the limits of $S_\ell'(\alpha)$ at $0$ and $+\infty$. Easily, $$\lim_{\alpha\rightarrow+\infty}S_\ell'(\alpha)=\lim_{\alpha\rightarrow+\infty}Q({\ell}/{\alpha}) = {1}/{2}.$$ For the limit $\alpha\rightarrow 0^+$, use the following facts: (i) in the limit of $x\rightarrow+\infty$: $Q(x)\sim p(x)/x$, and, (ii) $\lim_{x\rightarrow+\infty}xe^{-x^2/2}=0$, to conclude with the desired.
\end{proof}


Observe that $F(\tau)$ in \eqref{eq:F} can be written as
\begin{align}\label{eq:rewrite_F}
F(\tau) = \tau(\delta-\frac{1}{2})-\frac{1/SNR}{\tau}+S(\tau;2).
\end{align}
 We are now ready to state and prove Lemma \ref{lem:tau_prop}.

\begin{lem}\label{lem:tau_prop}[Properties of $\tau_*$]
Let $\tau_*$ be defined as in Theorem \ref{thm:main}, i.e. the (unique) positive minimizer of the function $F(\tau)$ in \eqref{eq:F}.  The following hold.
\begin{enumerate}[(a)]
\item $\tau_*$ is the unique positive solution of the equation 
\begin{align}\label{eq:fo_tau}
\delta-\frac{1}{2} - \frac{1/\SNR}{\tau^2} + G(\tau^{-1}) = 0,
\end{align} 
where 
\begin{align}\label{eq:Gu}
G(u):= {\sqrt{({2}/{\pi}})}ue^{-2u^2} + (1- {4}u^2)\cdot Q({2}u).
\end{align}

\item $\tau_*$ satisfies \eqref{eq:bounds_tau}.

\end{enumerate}

\end{lem}

\begin{proof}
Recall from Theorem \ref{thm:main} that the function $F(\tau)$ in \eqref{eq:F} is strictly convex. Hence, $\tau_*$ is the unique positive solution to the first-order optimality condition: $F'(\tau):=\frac{\mathrm{d}}{\mathrm{d}\tau}F(\tau) = 0.$ 
It is convenient for the rest of the proof to define a function $H:(0,\infty)\rightarrow\R$ as follows: $$H(u):=F'(u^{-1}).$$
Also, note from \eqref{eq:S'} that $G$ in \eqref{eq:Gu} satisfies 
\begin{align}
\label{eq:G=S}
G(u)=S_2'(u^{-1}). 
\end{align}
In particular,  properties of $G$ to be used later in the proof follow from Lemma \ref{lem:technical}.

Starting with \eqref{eq:rewrite_F} and using Lemma \ref{lem:technical}(a) and \eqref{eq:G=S}:
$$
H(u):= \delta-\frac{1}{2} - \frac{u^2}{\SNR}   + \underbrace{{\sqrt{\frac{2}{\pi}}}ue^{-2u^2} + (1- {4}u^2) Q({2}u)}_{:=G(u)}.
$$
This proves the first statement. Moreover, since $F(\tau)$ is strictly convex, we have that $F'(\tau)$ is strictly increasing, and equivalently that $H(u)$ is a decreasing function of $u$. 

Next, we prove that, 
\begin{align}\label{eq:high_up}
\tau_*^{-1} \geq \sqrt{(\delta-{1}/{2})\SNR}=:\tau_0^{-1}.
\end{align}
From Lemma \ref{lem:technical}(c) and \eqref{eq:G=S}, 
$$G(u)>0, \quad \text{for all } u> 0.$$ Hence,
$H(\tau_0^{-1}) = G(\tau_0^{-1}) >0.$
But, $H(u)$ is decreasing and $\tau_*^{-1}$ is its unique zero, from which \eqref{eq:high_up} follows.

Finally, we show that
\begin{align}\label{eq:high_low}
\tau_*^{-1} < \sqrt{\delta\cdot\SNR}:=\tau_1^{-1}.
\end{align}
Note that,
$$H(\tau_1^{-1}) = -\frac{1}{2} + G(\tau_1^{-1}).$$
Again, from Lemma \ref{lem:technical}(c) and \eqref{eq:G=S}, it follows that $G(u)<1/2$. Therefore, $H(\tau_1^{-1})< 0$.
Combine this with the fact that $H(u)$ is decreasing and $\tau_*^{-1}$ is its unique zero, to conclude with \eqref{eq:high_low}, as desired.\end{proof}

\subsubsection{High-SNR regime}\label{sec:high-SNR}

Theorem \ref{thm:high-SNR} below formalizes and proves \eqref{eq:BRO_high}. 

\begin{thm}[High-SNR regime]\label{thm:high-SNR}
As in the statement of Theorem \ref{thm:main}, fix $\delta\in(\frac{1}{2},\infty)$ and let $\BER$ denote the bit error probability of the detection scheme in \eqref{eq:algo} for some fixed but unknown BPSK signal $\x_0\in\{\pm1\}^n$. For any $\eps>0$, there exists constant $\overline{\SNR}:=\overline{\SNR}(\epsilon)$ such that for all values $\SNR>\overline{\SNR}$, it holds
$$
\lim_{\substack{m,n\rightarrow\infty \\ {m}/{n}\rightarrow\delta}} \Pro\big( \Big| \frac{\BER}{Q\left(\sqrt{\left(\delta-{1}/{2}\right)\SNR}\right)} - 1 \Big| > \eps \big) = 0.
$$
\end{thm}

\begin{proof}
Fix any $\eps>0$. Recall $\tau_*:=\tau_*(\SNR)$, the minimizer of \eqref{eq:F}, and define for convenience:
\begin{align}\label{eq:tau0}
\tau_0:=\tau_0(\SNR) = \left(\sqrt{(\delta-{1}/{2})\SNR}\right)^{-1}.
\end{align}
 We will prove that there exists $\overline{\SNR}(\epsilon)$, such that 
\begin{align}\label{eq:high_2show}
\Big|\frac{Q\left({\tau_*}^{-1}\right)}{Q\left({\tau_0}^{-1}\right)} - 1\Big|\leq \frac{\eps}{2},
\end{align}
for all $\SNR\geq\overline{\SNR}(\epsilon)$. 
 This would suffice to complete the proof of the theorem. To see this, write
\begin{align}
\big|\frac{\BER}{Q({\tau_0}^{-1})} - 1\big| &= \big|\frac{\BER-Q({\tau_*}^{-1})}{Q({\tau_0}^{-1})} + \frac{Q({\tau_*}^{-1})}{Q({\tau_0}^{-1})} - 1\big| \nn\\
&\leq \frac{|\BER-Q({\tau_*}^{-1})|}{Q({\tau_0}^{-1})} + \big|\frac{Q({\tau_*}^{-1})}{Q({\tau_0}^{-1})} - 1\big|\nn,
\end{align}
and observe the following. (a) The last term above is further upper bounded by $\eps/2$ using \eqref{eq:high_2show} for large enough $\SNR>\overline{\SNR}(\eps)$. (b) From Theorem \ref{thm:main}, for all values of $\SNR$, there exist large enough $m,n$ such that the nominator of the first term is upper bounded by $(\eps/2)Q({\tau_0}^{-1})$ with probability 1.


In what follows, we show \eqref{eq:high_2show}, which is a deterministic statement about the minimizer $\tau_*:=\tau_*(\SNR)$ of \eqref{eq:F}. We use Lemma \ref{lem:tau_prop}. 

From \eqref{eq:bounds_tau}, we have that
\begin{align}\label{eq:high_up_lim}
\lim_{\SNR\rightarrow+\infty} \tau_*^{-1}  = +\infty.
\end{align}
Also, recall from \eqref{eq:tau0} that $(\delta-1/2)=\frac{\tau_0^{-2}}{\SNR}$. Substituting this in \eqref{eq:fo_tau} we find that
\begin{align}\label{eq:high_sand1}
0\leq \tau_*^{-2} - \tau_0^{-2} = \SNR\cdot G(\tau_*^{-1})
\end{align}
for $G$ as in \eqref{eq:Gu} (also, recall \eqref{eq:G=S}).
The non-negativity above follows from the lower bound in \eqref{eq:bounds_tau}.
From Lemma \ref{lem:technical}(c) and \eqref{eq:G=S}, $G$ is decreasing in $(0,\infty)$. Using this, and applying the  lower bound in \eqref{eq:bounds_tau} once more, \eqref{eq:high_sand1} leads to the following:
\begin{align}\label{eq:high_sand2}
0\leq \tau_*^{-2} - \tau_0^{-2} \leq \SNR\cdot G(\tau_0^{-1}) = \SNR\cdot G(\sqrt{(\delta-1/2)\SNR}). 
\end{align}
But, from Lemma \ref{lem:technical}(c) the limit of the right-hand side as $\SNR\rightarrow+\infty$ is equal to $0$.
Combining,
\begin{align}\label{eq:high_lim_2}
\lim_{\SNR\rightarrow+\infty} ( \tau_*^{-2} - \tau_0^{-2} ) = 0.
\end{align}
Next, write $\tau_*^{-2} - \tau_0^{-2} = \tau_*^{-2}( 1- \frac{\tau_*^{2}}{\tau_0^{2}})$ and combine \eqref{eq:high_up_lim} with \eqref{eq:high_lim_2} to further show that 
\begin{align}\label{eq:high_lim_3}
\lim_{\SNR\rightarrow+\infty} \frac{\tau_*}{\tau_0} = 1.
\end{align}

We are now ready to prove \eqref{eq:high_2show}. For simplicity, we write $f(x)\sim g(x)$ instead of $\lim_{x\rightarrow+\infty}\frac{f(x)}{g(x)}=1$. It is well known that $Q(x)\sim p(x)/x$. Therefore,
\begin{align*}
\frac{Q(\tau_*^{-1})}{Q(\tau_0^{-1})}  &\sim \frac{p(\tau_*^{-1})}{p(\tau_0^{-1})} \frac{\tau_0}{\tau_*} =   \frac{\tau_0}{\tau_*} \exp\left(-\frac{\tau_*^{-2}-\tau_0^{-2}}{2}\right)\\
&\sim 1,
\end{align*}
where the second line follows from \eqref{eq:high_lim_2} and \eqref{eq:high_lim_3}.
\end{proof}

\subsection{Supplementary proofs for Section \ref{sec:proof}}\label{sec:proof_proof}

\subsubsection{From Lipschitz to the indicator function}

\begin{lem}[Approximating the indicator]\label{lem:ind}
Let $\mu$ be a continuous measure on the real line such that $c\in\R$ is a  point of measure zero. Further let $\{\mu_n\}$ be a sequence of random measures indexed by $n$ such that as $n\rightarrow +\infty$,
$$
\int \psi \mathrm{d}\mu_n \rP \int \psi \mathrm{d}\mu,
$$
for all Lipschitz functions $\psi:\R\rightarrow\R$. For the indicator function $\chi_c(\alpha):=\ind{\alpha\leq c}$ it holds that,
$$
\int \chi_c \mathrm{d}\mu_n \rP \int \chi_c \mathrm{d}\mu.
$$
\end{lem}
\begin{proof}
Fix any $\eps,\zeta>0$ and consider the random variable
$
X = \big| \int \chi_c \mathrm{d}\mu_n - \int \chi_c \mathrm{d}\mu \big|.
$
Note that  is random since the measures $\mu_n$ are random.  It will suffice to show that there exists $N_*$ such that for all $n>N_*$: $\Pro(X>\eps)\leq \zeta$.

Let $\eta>0$, the exact value of which to be determined later, and, consider the following functions parametrized by $\eta$:
$$
\psiu_{\eta}(\alpha):=\begin{cases} 
1, & \alpha\leq c \\
 1-\frac{1}{\eta}(\alpha-c), & c\leq\alpha\leq c+\eta\\
 0, & \alpha\geq c+\eta,
  \end{cases}
$$
and
$$
\psil_{\eta}(\alpha):=\begin{cases} 
1, & \alpha\leq c-\eta \\
 -\frac{1}{\eta}(\alpha-c), & c-\eta\leq\alpha\leq c\\
 0, & \alpha\geq c.
  \end{cases}
$$
These functions are both Lipschitz with Lipschitz constant $1/\eta$. Define, the random variable $Y_\eta$ as
\begin{align*}
Y_\eta:=\max\{\big| \int \psiu_\eta \mathrm{d}\mu_n - \int \psiu_\eta \mathrm{d}\mu \big|~,~
\big| \int \psil_\eta \mathrm{d}\mu_n - \int \psil_\eta \mathrm{d}\mu \big|\}.
\end{align*}
From the assumption of the lemma there is $N(\eps,\zeta,\eta)$ such that for all $n\geq N(\eps,\zeta,\eta)$:
\begin{align}\label{eq:P_1}
\Pro( Y_\eta > {\eps}/{2} ) \leq {\zeta}.
\end{align}

 Moreover, $\psil_\eta(\alpha)\leq \chi_c(\alpha)\leq \psiu_\eta(\alpha)$. Thus,
\begin{align}\label{eq:P_2}
X \leq Y_\eta + \int |\psiu_\eta-\psil_\eta| \mathrm{d}\mu \leq Y_\eta + \mu\{[c-\eta,c+\eta]\},
\end{align}
where for the second inequality we further used the fact that $|\psiu_\eta-\psil_\eta|$ is upper bounded by $1$ and has support $[c-\eta,c+\eta]$. 

Finally, from continuity of $\mu$ and the fact that $c$ is $\mu$-measure zero, we can choose $\eta=\eta_*(\eps)$ such that 
\begin{align}\label{eq:P_3}
\mu\{[c-\eta,c+\eta]\} \leq {\eps}/{2}.
\end{align}

Combining, \eqref{eq:P_1}--\eqref{eq:P_3}, we conclude, as desired, that there is $N_*:=N(\eps,\zeta,\eta_*(\eps))$ such that for all $n>N_*$ it holds
$$
\Pro( X > \eps ) \leq \Pro( Y_\eta > {\eps}/{2} ) \leq \zeta.
$$
\end{proof}

\subsubsection{Proof of Corollary \ref{cor:indep}}\label{sec:indep_proof}
On the one hand, by Theorem \ref{thm:general}(b), it holds for all $\ell=1,\ldots,k$ that
$$
\overline{\psi_\ell}(\wh) := n^{-1}\sum_{i=1}^n{\psi_\ell(\wh_i)}\rP\E_{W_\ell}[\psi_\ell(W_\ell)].
$$
On the other hand, for some constant $C>0$
$$
\Big|\prod_{\ell=1}^k\overline{\psi_\ell}(\wh) - n^{-k}\sum_{1\leq i_1,\ldots,i_k\leq n} \psi_1(\hat\w_{i_1})\cdots\psi_k(\hat\w_{i_k})\Big|\leq \frac{C}{n}.
$$
To see this, expand the product term on the left-hand side and use the boundedness of the functions $\psi_\ell$.

Combining the above proves the first statement of the corollary. The second statement follows with the exact same argument starting from Theorem \ref{thm:main} and observing that 
$\ind{\hat\w_{i_1}\leq-1,\ldots,\hat\w_{i_k}\leq-1}=\prod_{\ell=1}^k{\ind{\hat\w_{i_\ell}}}$.

\vp
\subsubsection{Proof of Lemma \ref{lem:strong_AO}}\label{sec:proof_strong_AO}
Denote, $\psio:=\E_W[\psi(W)]$. From Lemma \ref{lem:Lip_AO}, it holds w.p.a.1:
$
|{\frac{1}{n}\sum_{i=1}^{n}\psi(\wt_i)} - \psio| \leq \eps/2.
$
Hence, by definition of the set $\Sc$ and the triangle inequality, it holds w.p.a.1 that for all $\w\in\Scc$:
$
|{\frac{1}{n}\sum_{i=1}^{n}\psi(\w_i)} - {\frac{1}{n}\sum_{i=1}^{n}\psi(\wt_i)}| \geq \eps/2.
$
Then, the Lipschitz property of $\psi$ guarantees that 
\begin{align}\label{eq:w_div}
\frac{\| \w - \wt \|}{\sqrt{n}} \geq \frac{\eps}{2L}.
\end{align}

In what follows we show that $n\cdot f(\w)$ is $C$-strongly convex for appropriate constant $C>0$. In view of \eqref{eq:w_div} and recalling $\phi(\g,\h)=f(\wt)$, this will suffice to complete the proof.

It can be checked that the Hessian $\nabla^2f(\w)$ satisfies $n\nabla^2f(\w)\succcurlyeq\frac{\|\g\|_2}{\sqrt{n}}\frac{\sigma^2}{\sqrt{\frac{\|\w\|^2}{n}+\sigma^2}}\mathbf{I}.$
Further use the fact that ${\|\g\|_2}/\sqrt{n}\geq\sqrt{\delta}/2$ w.p.a.1 and $\|\w\|^2\leq 4n$, to conclude that w.p.a.1 $F$ is $\frac{C}{n}$-strongly convex with $C:=\frac{\sigma^2\sqrt{\delta}}{2\sqrt{\sigma^2+4}},$ or
$f(\w)\geq f(\wt) + \frac{C}{2}\frac{\|\w-\wt\|}{\sqrt{n}}.$

\subsection{Proof of Theorem \ref{thm:BER_PAM}}\label{sec:proof_PAM}
The proof of the theorem requires repeating, mutatis mutandis, the line of arguments detailed in Section \ref{sec:proof} for the proof of Theorem \ref{thm:main}. We omit most of the details for brevity, and only show the necessary calculations that yield to function $F_M$ in \eqref{eq:FM}.
The idea is the same as in Section \ref{sec:proof}: thanks to the CGMT, it suffices to analyze a corresponding Auxiliary Optimization (AO) instead of the original optimization in \eqref{eq:LASSO_MPAM}. Repeating the steps in Section \ref{sec:proof_simplify}, the corresponding (AO) becomes (compare to Eqn.~\eqref{eq:AOw25}): 
\begin{align}
\min_{\substack{\tau\geq 0}}\frac{\tau\|\g\|}{2\sqrt{n}} + \frac{\sigg\|\g\|}{2\tau\sqrt{n}}+ \frac{1}{n}\sum_{i=1}^n \min_{\ellix \leq\w_i\leq \uix} \left\{ \frac{\gm}{2\tau \sqrt{n}}\w_i^2-{\h_i}\w_i\right\},\nn
\end{align}
where, as always $\w=\x_0-\x$ denotes the ``error-vector" and we further defined 
\begin{align*}
\ellix&:=-(M-1)-\x_{0,i}~~\text{and}~~\uix&:=(M-1)-\x_{0,i}.
\end{align*}
 For simplicity in notation, further denote $A=\frac{\|\g\|}{\tilde{\tau}\sqrt{n}}$. 
Then, 
the optimal $\tilde\w_i:=\tilde\w_i(\g,\h,\x_0)$ satisfies
\begin{align}\label{eq:BER_w_opt_PAM}
\tilde\w_i = \begin{cases}
\ellix &,\text{if } \h_i< A\ellix,\\
\frac{1}{A}\h_i &,\text{if } A\ellix \leq \h_i \leq A\uix, \\
\uix &,\text{if }  \h_i> A\uix.
\end{cases}
\end{align}
where, $\tilde{\tau}:=\tilde{\tau}(\g,\h,\x_0)$ is the solution to the following:
%
\begin{align}\label{eq:BER_AOw3_PAM}
\Big(\min_{\tau>0} \frac{\tau\|\g\|}{2\sqrt{n}} + \frac{\sigg\gm}{2\tau\sqrt{n}} + \frac{1}{{n}}\sum_{i=1}^n \upsilon_{n}\Big(\frac{\tilde{\tau}\sqrt{n}}{\|\g\|};\h_i,\ellix,\uix\Big) \Big)_+, 
\end{align}
with
\begin{align}
\upsilon_{n}(\alpha;h,\ell,u) :=\begin{cases}
\frac{1}{2\alpha}\ell^2 - h\ell &,\text{if } \alpha h< \ell,\\
-\frac{\alpha}{2}h^2 &,\text{if } \ell \leq \alpha h \leq u, \\
\frac{1}{2\alpha}u^2 - hu &,\text{if }  \alpha h> u.
\end{cases}\nn
\end{align}

This is of course very similar to Equation \eqref{eq:AOw3}. Next, we follow the same steps as in Section \ref{sec:proof_4} and study the convergence of the (AO) in \eqref{eq:BER_AOw3_PAM}. For the first two summands in \eqref{eq:BER_AOw3_PAM}, we use the fact that  $\frac{\gm}{\sqrt{n}}\rP \sqrt\delta$. For the third summand, recall that each $\x_{0,i}$ takes values $\pm1,\pm3,\ldots,\pm(M-1)$ with equal probability $1/M$. Let $j=1,3,\ldots,M-1$ and denote,
$$
\elli:=(M-1)-j \quad\text{and}\quad \ui:=(M-1)+j.
$$
Then, the pairs $(\ellix,\uix)$ take values $(-\ui,\elli)$ and $(-\elli,\ui)$ with equal probability $1/M$ each. With these, 
 $
\frac{1}{{n}}\sum_{i=1}^n \upsilon_{n}\left(\frac{\tau\sqrt{n}}{\|\g\|};\h_i,\ellix,\uix\right)\rP
Y\left(\frac{\tau}{\sqrt{\delta}}\right),
$
where
%
%
\begin{align}\nn
&Y(\alpha) := \frac{1}{M}\sum_{j=1,3,\ldots,M-1} \E_{h\sim\Nn(0,1)}[ \upsilon_n(\alpha;h,-\ui,\elli)  ] \\\quad&+ \frac{1}{M}\sum_{j=1,3,\ldots,M-1} \E_{h\sim\Nn(0,1)}[ \upsilon_n(\alpha;h,-\elli,\ui)  ]. \label{eq:Y_M}
\end{align}
%
Simple calculations show that
\begin{align*}
&\E_{h\sim\Nn(0,1)}[ \upsilon_n(\alpha;h,\ell,u) ] =  \nn\\
&\quad-\frac{\alpha}{2} +  \frac{\alpha}{2} \int_{\frac{\ell}{\alpha}}^{\infty}(h-\frac{\ell}{\alpha})^2p(h)\mathrm{d}h + 
\frac{\alpha}{2} \int_{\frac{u}{\alpha}}^{\infty}(h-\frac{u}{\alpha})^2p(h)\mathrm{d}h.
\end{align*}
For convenience, define (see also Lemma \ref{lem:technical})
\begin{align}
S(\alpha;\ell) &:= \alpha\int_{\frac{\ell}{\alpha}}^{\infty}(h-\frac{\ell}{\alpha})^2p(h)\mathrm{d}h \nn \\
&=\left(\alpha + \frac{\ell^2}{\alpha}\right)Q\Big(\frac{\ell}{\alpha}\Big) - \frac{1}{\sqrt{2\pi}}\ell e^{-\frac{\ell^2}{2\alpha^2}}.
\end{align}
Putting all these together with \eqref{eq:Y_M} and grouping terms we find that
\begin{align*}
&Y(\alpha) = \frac{1}{M}\sum_{j=1,3,\ldots,M-3}
\Big(-\alpha + S(\alpha;\elli) + S(\alpha;\ui) \Big)
\nn
\\
&+ 
 \frac{1}{M}\left( - \frac{\alpha}{2} + S(\alpha;u_{M-1}) \right) =-\frac{\alpha}{2}\left(\frac{M-1}{M}\right) \\&+ \frac{1}{M} \sum_{j=1,3,\ldots,M-3}\left\{ S(\alpha;\ell_j) + S(\alpha;u_j) \right\} + \frac{1}{M} S(\alpha;u_{M-1}).
\end{align*}
Observe that $Y(\alpha)$ is nonnegative for $\alpha>0$ as long as  $\delta>\frac{M-1}{M}$. Therefore, we can repeat the technical arguments of Section \ref{sec:proof_4}, to conclude that the random optimization in \eqref{eq:BER_AOw3_PAM} converges to the following deterministic optimization (where, for convenience, we have rescaled the optimization variable $\tau$ as follows $\tau:=\frac{\tau}{\sqrt{\delta}}$):
\begin{align}\label{eq:obj_M}
\min_{\tau>0}\frac{\tau\delta}{2}+\frac{\sigg}{2\tau}+Y(\tau).
\end{align}
The objective function in \eqref{eq:obj_M} can be identified with the function $\FM(\tau)$ in the statement of the theorem. From Lemma \ref{lem:technical}(b) the second derivative of $\FM(\tau)$ is strictly positive for $\tau>0$, hence \eqref{eq:obj_M} has a unique minimizer, which we denote $\tau_*$. With arguments same as in the end of Section \ref{sec:proof_4}, we can show that $\sqrt{\delta}\tilde\tau(\g,\h,\x_0)\rP\tau_*$. 

Finally, we sketch how all these leads to the desired, namely:
$$\frac{1}{n}\sum_{i=1}^n\ind{\x^*_i\neq \x_{0,i}} \rP 2\Big(1-\frac{1}{M}\Big) Q(\tau_*^{-1}).$$
First, consider the case: $\x_{0,i}\in\{\pm1,\pm3,\ldots,\pm(M-3)\}$. Then, the thresholding rule \eqref{eq:LASSO_thr} implies that there is an error iff $|\tilde\w_i|>1$. Equivalently, in view of \eqref{eq:BER_w_opt_PAM}, and noting that $\uix\geq 2$, it follows that and error occurs iff $|\h_i|>A$.  Next, consider the case(s) $\x_{0,i}=M-1$ (or, $\x_{0,i}=-(M-1)$). Then the error event corresponds to $\tilde\w_i<-1$ (or, $\tilde\w_i>1$), which in view of \eqref{eq:BER_w_opt_PAM} translates to $\h_i<-A$ (or $\h_i>A$). Putting these together and conditioning on the high-probability events $\|\g\|/\sqrt{n}\rP\sqrt{\delta}$ and $\tilde\tau\rP\tau_*$, we find that
\begin{align*}
&\frac{1}{n}\sum_{i=1}^n\ind{\arg\min_{s\in\Cc}|\x_{0,i}+\tilde\w_i-s|\neq \x_{0,i}} \rP \\ & \frac{2}{M}\left( (M-2) Q(\tau_*^{-1}) + Q(\tau_*^{-1}) \right) = 2\Big(1-\frac{1}{M}\Big) Q(\tau_*^{-1}).
\end{align*}

\end{document}